\documentclass[10pt,leqno]{amsart}
\usepackage{graphicx}
\usepackage{amssymb}
\usepackage{indentfirst,csquotes}
\usepackage{mathtools}
\usepackage{graphicx} 
\usepackage{stmaryrd} 
\usepackage{color}
\usepackage{mathrsfs}
\usepackage{eucal}
\usepackage{pgfkeys}
\usepackage{yfonts}
\usepackage{comment}
\usepackage{tikz}
\usepackage{pdfsync}
\usepackage{epsfig}
\usepackage{graphicx, subfigure}
\usepackage{float}

\def\i{\bf{\rm i}}

\usepackage{amssymb,amsthm, amsmath}
\usepackage{xcolor,paralist,hyperref,titlesec,fancyhdr,etoolbox}
\newtheorem{thm}{Theorem}[]
\newtheorem{definition}[thm]{Definition}
\newtheorem{example}[thm]{Example}
\newtheorem{lemma}[thm]{Lemma}

\newtheorem{corollary}[thm]{Corollary}

\def \q{t}
\def \e{{\rm e}}

\def \i{{\mathrm i}}

\titleformat{\section}[block]%
  {\normalfont\Large\bfseries} 
  {\thesection.} 
  {0.5em} 
  {} 
\titlespacing*{\section}{0pt}{3ex}{2ex} 

\titleformat{\subsection}[block]{\normalfont\bfseries}{\thesubsection}{0.6em}{}

\newenvironment{customproof}[1][]
    {\noindent\textit{Proof of #1. }\ignorespaces}
    {\hfill $\Box$\vskip 5pt}

\usepackage{lipsum}

\begin{document}
\title{Triangular isomonodromic solutions to a Fuchsian system from superelliptic curves}

\author[Anwar Al Ghabra, Benjamin Piché, Vasilisa Shramchenko]{Anwar Al Ghabra$^{1}$, Benjamin Piché$^{1}$, Vasilisa Shramchenko$^{1}$}
\date{\today}
\maketitle

\footnotetext[1]{Department of mathematics, University of
	Sherbrooke, Sherbrooke, Quebec, Canada. E-mail: {\tt Mouhammed.Anwar.Al.Ghabra@Usherbrooke.ca},  {\tt Benjamin.Piche@USherbrooke.ca},  {\tt Vasilisa.Shramchenko@Usherbrooke.ca}}

\let\thefootnote\relax

\begin{abstract}
We give fundamental solutions of arbitrarily sized matrix Fuchsian linear systems, in the case where the coefficients $B^{(i)}$ of the systems are matrix solutions  of the Schlesinger system that are upper triangular, and whose eigenvalues follow an arithmetic progression of a rational difference. The values on the superdiagonals of the matrices $B^{(i)}$ are given by contour integrals of meromorphic differentials defined on Riemann surfaces obtained by compactification of superelliptic curves. We show that our fundamental solutions are isomonodromic by obtaining their monodromy matrices. 
\\

\hskip -0.4cm
Keywords: Fuchsian linear systems, Riemann-Hilbert problem, isomonodromic deformations, Schlesinger systems, Riemann surfaces, algebraic curves, superelliptic curves. 

\end{abstract} 

\bigskip


\noindent 
\section{Introduction}

For a fixed set of $N$ distinct points $a_1, \dots, a_N$ in the complex plane,  consider a matrix Fuchsian linear equation
\begin{equation}
\label{fuchsyst1}
\frac{\mathrm{d}\Phi}{\mathrm{d}z} = \sum_{i=1}^{N}\frac{B^{(i)}}{z-a_i}\Phi
\end{equation}
for a $p\times p$ matrix $\Phi(z)$ with $z\in\mathbb CP^1.$ Here the matrices $B^{(i)}$ are independent of $z$ but may depend on $a_1, \dots, a_N$. Generically,  solutions to this equation have singularities at $z=a_i$ and $z=\infty$ (unless $\sum_{i=1}^{N}{B^{(i)}}=0$) and, under analytical continuation along an element  of $\pi_1(\mathbb C\setminus\{a_1, \dots, a_N\}, z_0)$,  a fundamental solution $\Phi(z_0)$ transforms to another  fundamental  solution related to $\Phi(z_0)$ by the right multiplication with a constant matrix, called the {\it monodromy matrix} corresponding to the chosen loop. For fixed $a_1, \dots, a_N$,  monodromy matrices depend only on the homotopy class of the corresponding loops.  One thus obtains a {\it monodromy representation} of the fundamental group of the complex plane punctured at $a_i$, that is a homomorphism $\rho: \pi_1(\mathbb C\setminus\{a_1, \dots, a_N\}, z_0) \to GL(p,\mathbb C),$ see \cite{AnosovBolibruch1994}. 
The monodromy representation can also be seen as the holonomy  representation of the flat meromorphic connection on the trivial rank $p$ vector bundle over $\mathbb CP^1\setminus\{a_1, \dots, a_N, \infty\}$ whose connection form is $\sum_{i=1}^{N}\frac{B^{(i)}}{z-a_i}dz.$

\smallskip

Once a fundamental solution to \eqref{fuchsyst1} is found, one may obtain its monodromy matrices. 
The inverse problem, that of reconstructing a function $\Phi(z)$ from its monodromies, is referred to as a Riemann-Hilbert or inverse monodromy problem. Alternatively, it may be reformulated as a problem of finding two matrix functions, each of which is holomorphic either inside or outside of a closed contour passing through the points $a_1, \dots, a_N$  in the complex plane, such that the two functions are related by the multiplication with a ``jump'' matrix on the contour. The monodromy matrices associated with the loops around $a_i$ may be combined to obtain this jump matrix \cite{BolibruchRH,  Its}. 

\smallskip

Riemann-Hilbert problems are important in mathematics from the theoretical point of view, see the review \cite{Its}. At the same time, they play an important role in various applied mathematics contextes, see for example \cite{Baddoo, Kisil, Noble, Posson} and references therein. The term {\it Riemann-Hilbert problem} refers also to the 21st problem from Hilbert's list, that  of reconstructing a Fuchsian equation \eqref{fuchsyst1} from a collection of monodromy matrices associated with a set of points $a_1, \dots, a_N\in\mathbb C$,  \cite{BolibruchRH, AnosovBolibruch1994}, see also \cite{Rogo},  review \cite{GP} and references therein. 

\smallskip

Closely related to the Riemann-Hilbert problems is the search for isomonodromic deformations of Fuchsian systems, that is the question of finding the coefficients $B^{(i)}$ in  \eqref{fuchsyst1} for which there exist fundamental solutions of  \eqref{fuchsyst1} the monodromy matrices of which are independent of small variations of positions of the Fuchsian singularities, that  is variations of the set $\{a_1, \dots, a_N\}$. One way of finding such matrices $B^{(i)}$ is by solving the Schlesinger system \cite{Schlesinger1912, BolibruchFuchsian}, a system of nonlinear differential equations for $B^{(i)}$ with respect to the variables $a_1, \dots, a_N$ giving a sufficient condition for the resulting family of Fuchsian system to be isomonodromic, that is to share the same monodromies for $a=(a_1, \dots, a_N)$ in some set $D\subset \mathbb C^N$ not intersecting the diagonals of $\mathbb C^N$.

\smallskip

In the case $p=2$ and $N=3$, setting by a M\"obius transformation $\{a_1, a_2, a_3\}=\{0,1,x\}$,  the Schlesinger system reduces to a Painlev\'e-VI equation with parameters related to the eigenvalues of the matrices $B^{(i)}$, see for example \cite{Guzzetti}.

\smallskip

Solutions to the Schlesinger system are constructed in \cite{DIKZ, KK, DS1, DS0, DS2}  for $p=2$, in \cite{Go, GL} for $p=2$ and $p=3$ and
\cite {EG, Kor, DragGontShram}  for an arbitrary $p$. Most of these solutions are constructed in terms of functions defined on Riemann surfaces and over families of Riemann surfaces. The surfaces are obtained 
by considering coverings of $\mathbb CP^1$ ramified over the set of Fuchsian singularities of the associated system \eqref{fuchsyst1}, that is over the set $\{a_1, \dots, a_N\}$ and possibly the point at infinity. Allowing the {\it branching set} to vary, one considers a family of Riemann surfaces naturally associated with system \eqref{fuchsyst1} and thus with the Schlesinger system. 
 
 \smallskip
 
Solving the Fuchsian system \eqref{fuchsyst1} having coefficients $B^{(i)}$ obtained as solutions to the Schlesinger system in the algebro-geometric form in the above references is typically quite difficult. One exception being \cite{Kor} where a Riemann-Hilbert problem with regular singularities was solved for quasi-permutation monodromy matrices (independent of positions of regular singularities), after which the corresponding solution to the Schlesinger system was obtained. 

\smallskip

The goal of our paper is to solve systems \eqref{fuchsyst1} for the  coefficients $B^{(i)}$ obtained in \cite{DS2, DragGontShram}. 
The class of solutions of the Schlesinger system constructed in \cite{DS2, DragGontShram} are upper triangular and have a relatively simple form due to an additional assumption that the eigenvalues of each matrix $B^{(i)}$ form an arithmetic progression with a rational difference, the difference being the same for all $i=1, \dots, N.$ The eigenvalues of $B^{(i)}$ are called {\it exponents} of the associated Fuchsian system. Triangular solutions to Schlesinger systems are also investigated in \cite{DM2} and \cite{Go, GL}. 

\smallskip

We obtain fundamental solutions to systems \eqref{fuchsyst1} for all coefficients $B^{(i)}$ found in \cite{DragGontShram}. Our solutions are written in terms of integrals over closed contours of differentials defined on algebraic curves. The algebraic curves in question are presented as coverings of the Riemann sphere ramified over the set of singularities of the Fuchsian system, that is over the set $\{a_1, \dots, a_N, \infty\}$; the same curves were used in  \cite{DragGontShram} for construction of solutions $B^{(i)}$  to the Schlesinger system.  We show that the fundamental solutions we obtain are isomonodromic, that is have monodromy matrices independent of small variations of the positions of singularities of \eqref{fuchsyst1}. Moreover, their monodromy matrices may be written using a pattern similar to the one used to write expressions of the solutions themselves. This pattern involves sums over partitions of integers. 

\smallskip

Besides giving explicit algebro-geometric expressions for solutions to a class of Fuchsian systems,  our solutions might shed some light on the inverse monodromy correspondence for a larger class of monodromy data  due to the close relationship between the obtained expressions for monodromy matrices and for the corresponding fundamental solutions to \eqref{fuchsyst1}. 

\smallskip

The paper is organized as follows. In Section \ref{sect_setup} we recall the construction of matrices $B^{(i)}$ from \cite{DragGontShram} and the relevant algebraic curves. In Section \ref{sect_main} we present our main result giving solutions to the isomonodromic Fuchsian system with coefficients from \cite{DragGontShram}. In Section \ref{sect_monodromies} we obtain monodromy matrices of our solutions and show that they are independent of the positions of Fuchsian singularities. In Section \ref{sect_examples}, we consider two particular cases of integration contours leading to solutions of the Fuchsian system of the form $MD$ with a rational matrix $M$ and a diagonal matrix $D$ having a simple structure.

\section{Schlesinger system and superelliptic curves}
\label{sect_setup}

Let $B^{(1)}, B^{(2)}, \ldots, B^{(N)}$ be  $p\times p$ matrices depending on some complex parameters $a_1, \ldots, a_N$ defined on a ball  $D \subset \mathbb{C}^N \setminus \{(a_1, \ldots, a_N)\,\, | \,\, a_i = a_j \, \text{for some }\, i \neq j\}$. The Schlesinger system for matrices $B^{(j)}$ has the form
\begin{equation} \label{Schlesinger}
    \mathrm{d}B^{(i)} = - \sum_{j = 1, j \neq i}^{N} \frac{[B^{(i)}, B^{(j)}]}{a_i - a_j} \mathrm{d}(a_i - a_j),
\end{equation}
with $i = 1, \ldots, N.$

As a partial differential equations system, the Schlesinger system is given by
\begin{equation}
\label{introSchle}
\begin{dcases*}     
\frac{\partial B^{(i)}}{\partial a_j} = \frac{[B^{(i)}, B^{(j)}]}{a_i - a_j}, \quad  i \neq j; \\ 
\frac{\partial B^{(i)}}{\partial a_i} = - \sum_{j = 1, j\neq i}^{N}\frac{[B^{(i)}, B^{(j)}]}{a_i - a_j}. 
\end{dcases*} 
\end{equation}
The Schlesinger system was derived \cite{Schlesinger1912} as a condition of isomonodromy  for the Fuchsian system \eqref{fuchsyst1} for a matrix function $\Phi(z) \in M_p(\mathbb{C})$.
In other words, if the coefficients of system \eqref{fuchsyst1}, the matrices  $B^{(j)}\in M_p(\mathbb{C})$, satisfy system \eqref{introSchle} as functions of $a_1, \dots, a_N$, then the monodromy matrices of solutions $\Phi$ of \eqref{fuchsyst1} on closed loops around $a_j$ in the $z$-plane do not change under small perturbations of positions of $a_1, \dots, a_N$. 

In \cite{DragGontShram}, some upper triangular matrix solutions to the Schlesinger system \eqref{introSchle} were found in terms of contour integrals  defined on the compact Riemann surfaces associated with superelliptic curves. In this paper, we give the corresponding upper triangular solutions of the Fuchsian system \eqref{fuchsyst1}.  The superelliptic curves from \cite{DragGontShram} are given by the equation
\begin{equation}
\label{curves}
\Gamma_a = \{(\zeta,w) \in \mathbb{C}^2 \, | \, w^m = \prod_{i=1}^N (\zeta-a_i)\}
\end{equation}
and we consider their family
parametrized by  $a = (a_1, \ldots, a_N)\in D$ with some fixed positive integers $m$ and $N$. The associated compact Riemann surfaces, denoted by $X_a$, are obtained by projectivizing and desingularizing the curves \eqref{curves}, see \cite{Kirwan92} and  see also \cite{DragGontShram,Piche24} for application to superelliptic curves. Note that the algebraic curves \eqref{curves} are nonsingular. The associated projective curves $\hat{\Gamma}_a \subset \mathbb{C}P^2$ may be singular at the points at infinity denoted collectively by $\{\infty\}$. 
By the theorem of resolution of singularities (\cite{Kirwan92}, Chapter 7) there exists a holomorphic map $\pi: X_a \rightarrow \mathbb{C}P^2$ with  $\pi(X_a) = \hat{\Gamma}_a$ and such that the map
\begin{equation*}
    \pi: X_a \, \backslash \, \pi^{-1}(\{\infty\}) \rightarrow \hat{\Gamma}_a \, \backslash \, \{\infty\}
\end{equation*}
is biholomorphic. 
Given that the algebraic curves \eqref{curves} are non-singular, the map $\pi$ allows us to identify the points of the Riemann surface $X_a$ with those of the  curve $\Gamma_a$  as follows. We say that $P=(\zeta, w)$ is a point of the surface $X_a$ if $P=\pi^{-1}([\zeta:w:1]).$ For the points at infinity on $X_a$ we use the notation $\{\infty_1, \dots, \infty_s\}:=\pi^{-1}(\{\infty\})$. Note that the number of points at infinity is the greatest common divisor of $N$ and $m$, that is $s={\rm gcd}(m,N)$, and the genus of $X_a$ is $ g(X_a)=\frac12\bigl((m-1)(N-1)-s+1\bigr)$
 \cite{DragGontShram}.

The projection to the first coordinate $\zeta$ is a holomorphic function on the algebraic curve $\zeta:\Gamma_a \to \mathbb C$. It extends naturally to a meromorphic function on the Riemann surface $\zeta:X_a \to \mathbb CP^1$ by setting $\zeta(\zeta,w)=\zeta$ and $\zeta(\infty_j)=\infty$ with a slight abuse of notation. This way we can identify the Riemann surface $X_a$  with the surface of the ramified covering $(X_a, \zeta)$ of degree $m$ ramified over the points $a_1, \dots, a_N, \infty$ called the branch points of the covering. We denote the corresponding ramification points by $P_{a_i}=(a_i, 0)$ for $i=1, \dots, N$ and by $\infty_j$ as above for $j=1,\dots, s.$

The structure of the ramified covering $\zeta:X_a \to \mathbb CP^1$ induces the {\it standard local coordinates} on the surface $X_a$ as follows \cite{DragGontShram}:
\begin{align}
\label{coordinates}
&\xi_{a_k}(P)=(\zeta(P) - a_k)^{\frac{1}{m}}, &&\mbox{if}\quad P\sim P_{a_k}, \quad\mbox{with}\quad k=1, \dots, N\,,
\nonumber
\\
& \xi_{\infty_j}(P)= (\zeta(P))^{-\frac{1}{m_1}}, &&\mbox{if}\quad P\sim \infty_j, \quad\mbox{ where } m=sm_1.
\\
& \xi_Q(P)=\zeta(P)-\zeta(Q), &&\mbox{if}\quad P\sim Q \quad\mbox{and $Q$ is a regular point.}
\nonumber
\end{align}
We thus consider the family of compact {\it superelliptic} Riemann surfaces $X_a$, $a\in D,$ associated with the family of algebraic curves \eqref{curves}. Functions and differentials on the surfaces $X_a$ may also be seen as functions over the family of surfaces, that is functions of $a=(a_1, \dots, a_N).$

The solutions from  \cite{DragGontShram} are constructed using the following meromorphic differentials on $X_a$ defined for any non-zero integer $n$ coprime to $m$ of \eqref{curves}:
\begin{equation}\label{diffmero}
    \Omega_{i}^{(j)}(a) := \frac{w^{jn} \mathrm{d}\zeta}{\zeta - a_{i}}, \quad \text{with}\quad j = 1, \ldots, p-1, \quad\text{and}\quad i=1, \dots, N.
\end{equation}
Recall that $\zeta$ is a meromorphic function on $X_a$ discussed above. In the same way, $w$ is a meromorphic function on $X_a$ extended from $w:\Gamma_a\to \mathbb C$ by setting $w(\infty_k)=\infty$ for $k=1, \dots, s.$
\begin{thm}
\label{thmDGS}
\cite{DragGontShram}
Let $X_a$ with $a=(a_1, \dots, a_N)\in D$ be a family of compact Riemann surfaces defined above corresponding to the family of algebraic curves \eqref{curves}. 
Let $m$ and $N$ be  natural integers from \eqref{curves} and let $n \in \mathbb{Z}$ be  coprime to $m$. Let $\Omega_{i}^{(j)}(a)$ be the meromorphic differential defined by \eqref{diffmero} on $X_a$. Then a set of upper triangular $p\times p$ matrices $B^{(i)}$, $i = 1, . . . , N$ defined by the following three conditions provides a solution to the Schlesinger system \eqref{introSchle}:
\begin{itemize}
 \item $B^{(i)}_{kl}=0$ if $k>l;$
\item
The eigenvalues of each matrix $B^{(i)}$, $i = 1, . . . , N$, satisfy 
\begin{equation}
\label{condetoile}
B^{(i)}_{kk} -B^{(i)}_{k+1,k+1} = \frac{n}{m}, \quad \text{with } k = 1, \ldots , p - 1;
\end{equation} 
\item The entries on the superdiagonals are given by
\begin{equation}
\label{bi}
B^{(i)}_{kl} = \oint_{\gamma_{l - k}}  \Omega_{i}^{(l-k)}(a) = \oint_{\gamma_{l - k}} \frac{w^{(l-k)n} \mathrm{d}\zeta}{\zeta - a_{i}}, \quad l > k,
\end{equation}
 where $\gamma_1, \ldots, \gamma_{p-1}$ are linear combinations with complex coefficients of closed contours on the Riemann surface $X_a$ not passing through a pole of the corresponding $\Omega_{i}^{(j)}(a)$ and independent of small variations of $a\in D$. More precisely,  $\gamma_1, \ldots, \gamma_{p-1}$  are elements of
\begin{enumerate}
\item[\rm (a)] $H_1(X_a,\mathbb C)$ if $m$, $N$ are coprime and $n>0$,
\item[\rm (b)] $H_1(X_a\setminus\{\infty_1,\ldots,\infty_s\},\mathbb C)$ if $m$, $N$ are not coprime, $s={\rm gcd}(m,N)$, and $n>0$.
\item[\rm (c)]  $H_1(X_a\setminus\{P_{a_1},\ldots,P_{a_N}\},{\mathbb C})$ if $n<0\,.$
\end{enumerate}
\end{itemize}
\end{thm}
Note that for some contours, the integrals in \eqref{bi} vanish giving possibly trivial solutions to the Schlesinger system. If $m=1$ and $n>0$ then all the solutions are trivial as well, that is constant  diagonal matrices. If the contours are chosen to encircle poles of the respective $\Omega_{i}^{(j)}(a)$, solutions to the Schlesinger system become rational functions of $a_1, \dots, a_N$, see \cite{DragGontShram} and Section \ref{sect_examples} below. 

Solutions given by Theorem \ref{thmDGS} are such that $B^{(\infty)}:=-\sum_{i=1}^N B^{(i)}$ are non-zero constant diagonal matrices, see \cite{DragGontShram}. This implies that the connection form $\sum_{i=1}^{N}\frac{B^{(i)}}{z-a_i}dz$ has a simple pole at $z=\infty$ with residue $B^{(\infty)}.$

\section{Solutions of upper triangular Fuchsian systems}
\label{sect_main}

In order to simplify the presentation, let us introduce the following notation.
\begin{definition}
\label{dfnkappa}
Let $n \neq 0$ and $m > 0$ be coprime integers, $r = 1, \ldots, p-1$, where $p$ is an arbitrary natural number. 
 Let $z \in \mathbb C$ be fixed, and let $X_a$ be a Riemann surface associated with the superelliptic curve ${\Gamma}_a$ \eqref{curves}. Let $\zeta$ and $w = \prod_{i=1}^N(\zeta - a_i)^{1/m}$ be the meromorphic functions on $X_a$ as  in Section \ref{sect_setup}, and let $\gamma_r$ be a closed contour on $X_a$ as in Theorem \ref{thmDGS}. Assume in addition that $\gamma_r$ does not pass by the points $P$ of the surface for which $\zeta(P)=z.$   We then define
\begin{equation}
\label{kappa}
\kappa_r := -\frac{m}{rn}\oint_{\gamma_r}\frac{w^{rn}}{\zeta - z}\mathrm{d}\zeta. 
\end{equation}
\end{definition}
The derivative of $\kappa_r$ with respect to $z \in \mathbb{C}$, denoted by $\kappa'_r$, is obtained in a straightforward way:
\begin{equation}
\label{derivativekappa}
 \kappa'_r := \frac{\mathrm{d}}{\mathrm{d}z}\kappa_r = -\frac{m}{rn}\oint_{\gamma_r}\frac{w^{rn}}{(\zeta - z)^2}\mathrm{d}\zeta .
\end{equation}
In order to simplify further calculations, let us prove the following lemma.
\begin{lemma}
\label{propcalculsimp}
 Let $j$ be a strictly positive integer, $\gamma_j$  a closed contour on the Riemann surface $X_a$, and $\zeta,$ $w = \prod_{i=1}^N(\zeta - a_i)^{1/m}$ meromorphic functions on $X_a$ as in Section \ref{sect_setup}. For  $\kappa_j$ given by Definition \ref{dfnkappa} and for its derivative $\kappa'_j$  \eqref{derivativekappa} with respect to the complex parameter $z$, we have
\begin{equation*}
\frac{-jn}{m}\kappa_j\sum_{i=1}^N\frac{1}{z - a_i} + \kappa'_j  = \sum_{i=1}^N\frac{1}{z - a_i}\oint_{\gamma_j}\frac{w^{jn}}{\zeta - a_i}\mathrm{d}\zeta.
\end{equation*}
\end{lemma}
\begin{proof}
By Definition \ref{dfnkappa} and \eqref{derivativekappa}, we can write 
\begin{align}
\nonumber
\frac{-jn}{m}\kappa_j\sum_{i=1}^N\frac{1}{z - a_i} + \kappa'_j &=    \sum_{i=1}^N\left[\oint_{\gamma_j}\frac{w^{jn}\mathrm{d}\zeta}{\zeta - z} \frac{1}{z-a_i}\right] - \frac{m}{jn}\oint_{\gamma_j}\frac{w^{jn}}{(\zeta - z)^2}\mathrm{d}\zeta
\\
&=\sum_{i=1}^N\oint_{\gamma_j}\left(\frac{w^{jn}}{\zeta - z} + \frac{w^{jn}}{z - a_i}\right) \frac{\mathrm{d}\zeta}{\zeta - a_i} - \frac{m}{jn}\oint_{\gamma_j}\frac{w^{jn}}{(\zeta - z)^2}\mathrm{d}\zeta .
\label{lemmacalculsimp}
\end{align}
Due to the vanishing of the integral of an exact differential over a closed contour, for every $z$ we have
\begin{equation*}
 \oint_{\gamma_j} \mathrm{d}\frac{w^{jn}}{\zeta - z} = \oint_{\gamma_j} \mathrm{d}\frac{\prod_{i=1}^N(\zeta - a_i)^{jn/m}}{\zeta - z}
= \oint_{\gamma_j} \left(\frac{jn}{m}\frac{w^{jn}}{\zeta - z}\sum_{i=1}^N\frac{1}{\zeta - a_i} - \frac{w^{jn}}{(\zeta - z)^2} \right)\mathrm{d}\zeta=0,
\end{equation*}
which together with \eqref{lemmacalculsimp} concludes the proof.
\end{proof}

The next theorem states the main result of the present paper. 
\begin{thm}
\label{cjtpiche}
Let $p, \;n, \;m, \; N$ and $X_a, a\in D$  be as in Theorem \ref{thmDGS}. Let $z_0\in\mathbb C\setminus\{a_1, \dots, a_N\}$ and let $\gamma_1,\dots, \gamma_{p-1}$ be  closed contours on $X_a$ from Theorem \ref{thmDGS} not passing through any point $P\in X_a$ for which $\zeta(P)=z_0$. More precisely, let $\gamma_j$  from Theorem \ref{thmDGS} be seen as an element of
\begin{enumerate}
\item[\rm (a)] $H_1(X_a\setminus\{P\in X_a\,|\, \zeta(P)=z_0\},\mathbb C)$ if $m$, $N$ are coprime and $n>0$,
\item[\rm (b)] $H_1(X_a\setminus(\{\infty_1,\ldots,\infty_s\}\cup \{P\in X_a\,|\, \zeta(P)=z_0\}),\mathbb C)$ if $n>0$ and $s={\rm gcd}(m,N)>1$,
\item[\rm (c)]  $H_1(X_a\setminus(\{P_{a_1},\ldots,P_{a_N}\}\cup \{P\in X_a\,|\, \zeta(P)=z_0\}),{\mathbb C})$ if $n<0$,
\end{enumerate}
independent of small variations of $a\in D$. Let $\mathcal Z\subset\mathbb C\setminus\{a_1, \dots, a_N\}$ be a small neighbourhood of $z_0$ not intersecting projections on the $\zeta$-sphere of the contours $\gamma_1,\dots, \gamma_{p-1}$. 
For such a choice of the contour $\gamma_j$, let $\kappa_j$ be a function of $z\in\mathcal Z$ given by Definition \ref{dfnkappa}. 
Let $B^{(i)},\;i=1, \dots, N,$ be upper triangular $p \times p$ matrices given by Theorem \ref{thmDGS} with eigenvalues  
\begin{equation*}
\beta_k^{(i)}= B^{(i)}_{kk} = \theta_i - \frac{(k-1)n}{m}, \quad  k \in \{1, \ldots, p\}, 
\end{equation*}
where $\theta_i \in \mathbb{C}$. Then, for $z\in\mathcal Z$ a fundamental upper triangular matrix solution $\Phi(z)$ to the Fuchsian system \eqref{fuchsyst1} with coefficients being the matrices $B^{(i)}$, $i = 1, \ldots, N$, is given by
\begin{equation}
\label{productMD}
    \Phi = MD,
\end{equation}
where
\begin{equation}
\label{D}
    D = \mathrm{diag}\left(\prod_{i=1}^N(z-a_i)^{\beta_1^{(i)}}, \ldots, \prod_{i=1}^N(z-a_i)^{\beta_p^{(i)}}\right),
\end{equation}
and
\begin{equation}
\label{M}
    M_{kl} = \begin{dcases*}\sum_{q \, \vdash \, (l - k)}\,\prod_{j=1}^{l-k}\frac{\kappa_j^{\sigma_q(j)}}{\sigma_q(j)!} \quad &\text{when } $l>k$,\\ 
    \quad 1 \quad &\text{when } $l= k$,\\ 
    \quad 0 \quad &\text{when } $l < k$,
    \end{dcases*}
\end{equation}
where $q \vdash (l-k)$ represents the partitions of the integer $l - k$,  and $\sigma_q(j)$ is the number of repetitions of the part $j$ in the partition $q$.
\end{thm}
Note that as $z$ varies, the contours $\gamma_j$ may be adjusted in a continuous way to avoid the points $P\in X_a$ with $\zeta(P)=z$ while being elements of the homology $H_1(X_a\setminus(\{\infty_1,\ldots,\infty_s\}\cup\{P_{a_1},\ldots,P_{a_N}\}\cup \{P\in X_a\,|\, \zeta(P)=z\}))$.
In this way, starting from some $z_0\in\mathbb C\setminus\{a_1, \dots, a_N\}$, a set of contours from Theorem \ref{cjtpiche}, and the matrix $\Phi(z_0)$ \eqref{productMD}-\eqref{M}, we obtain $\Phi(z)$ for any $z\in\mathbb C\setminus\{a_1, \dots, a_N\}$ by analytical continuation of $\Phi(z_0)$.

Solutions given by Theorem \ref{cjtpiche} were first obtained in Theorems 1-3 of \cite{Ghabra22} for some special cases of rational residue matrices $B^{(i)}$ \eqref{bi} from \cite{DragGontShram}. Those cases are included among the solutions given by Corollaries \ref{cor_rat1} and \ref{cor_rat2} in Section \ref{sect_examples} below.

\begin{example}
To illustrate the statement of the theorem we spell out the solutions in the matrix form for $p=5$. 
For $i=1, \dots, N,$ the $5\times 5$ matrices $B^{(i)}$ given by Theorem \ref{thmDGS} for some constants $\theta_i\in\mathbb C$ is
\begin{equation}\label{bj5x5}
     B^{(i)} = \begin{pmatrix}
     \theta_i & \oint_{\gamma_1} \frac{w^{n}}{\zeta - a_i} \,\mathrm{d}\zeta & \oint_{\gamma_2} \frac{w^{2n}}{\zeta - a_i} \,\mathrm{d}\zeta & \oint_{\gamma_3} \frac{w^{3n}}{\zeta - a_i} \,\mathrm{d}\zeta & \oint_{\gamma_4} \frac{w^{4n}}{\zeta - a_i} \,\mathrm{d}\zeta\\ 
     0 &  \theta_i - \frac{n}{m} & \oint_{\gamma_1} \frac{w^{n}}{\zeta - a_i} \,\mathrm{d}\zeta & \oint_{\gamma_2} \frac{w^{2n}}{\zeta - a_i} \,\mathrm{d}\zeta & \oint_{\gamma_3} \frac{w^{3n}}{\zeta - a_i} \,\mathrm{d}\zeta \\
     0 & 0 & \theta_i - \frac{2n}{m} & \oint_{\gamma_1} \frac{w^{n}}{\zeta - a_i} \,\mathrm{d}\zeta & \oint_{\gamma_2} \frac{w^{2n}}{\zeta - a_i} \,\mathrm{d}\zeta \\
     0 & 0 & 0 & \theta_i - \frac{3n}{m} & \oint_{\gamma_1} \frac{w^{n}}{\zeta - a_i} \,\mathrm{d}\zeta \\
     0 & 0 & 0 & 0 & \theta_i -\frac{4n}{m}
 \end{pmatrix}.
\end{equation} 
Theorem \ref{cjtpiche} gives the following fundamental solution to the Fuchsian system \eqref{fuchsyst1} with the coefficients $B^{(i)}$ \eqref{bj5x5}:
\begin{equation*}
    \Phi(z) = MD,
\end{equation*}
where
\begin{equation*}
    M = \begin{pmatrix}
        1 & \psi_1 & \psi_2 & \psi_3 & \psi_4 \\
        0 & 1 & \psi_1 & \psi_2 & \psi_3 \\
        0 & 0 & 1 & \psi_1 & \psi_2 \\
        0 & 0 & 0 & 1 & \psi_1 \\
        0 & 0 & 0 & 0 & 1 & 
    \end{pmatrix} 
\end{equation*}
and
\begin{equation*}
    D = \mathrm{diag}\left(\prod_{i=1}^N(z - a_i)^{\theta_i}, \prod_{i=1}^N(z - a_i)^{\theta_i - \frac{n}{m}}, \ldots, \prod_{i=1}^N(z - a_i)^{\theta_i - \frac{4n}{m}} \right).
\end{equation*}
The entries $\psi_j$ of the matrix $M$ are the following functions of $z$ expressed in terms of $\kappa_j$ given by Definition \ref{dfnkappa}:
\begin{align*}
    \psi_1 &= \kappa_1,\\
    \psi_2 &= \kappa_2 + \frac{1}{2}\kappa_1^2,\\
    \psi_3 &= \kappa_3 + \kappa_1\kappa_2 + \frac{1}{6}\kappa_1^3,\\
    \psi_4 &= \kappa_4 + \kappa_1\kappa_3 + \frac{1}{2}\kappa_2^2 + \frac{1}{2}\kappa_1^2\kappa_2 + \frac{1}{24}\kappa_1^4.
\end{align*}
\end{example}

The contour $\gamma_j$ for each $j=1, \dots, p-1$ in Theorem \ref{cjtpiche} is the same as the respective contour $\gamma_j$ in Theorem \ref{thmDGS} used to construct the matrices $B^{(i)}$ serving as coefficients in \eqref{fuchsyst1}. However, each set of contours $\gamma_1, \dots, \gamma_{p-1}$ corresponding to a fixed solution $B^{(1)}, \dots, B^{(N)}$ \eqref{condetoile}, \eqref{bi} of the Schlesinger system and thus to a fixed Fuchsian system \eqref{fuchsyst1}, gives rise to several different solutions of the Fuchsian system (corresponding to different initial conditions for \eqref{fuchsyst1}). This happens because, as stated in Theorem \ref{cjtpiche}, the contours $\gamma_j$ are elements of the homology of the surface punctured at the points $P\in X_a$ for which $\zeta(P)=z_0.$ In other words, a given contour $\gamma_j$ in  $H_1(X_a\setminus\{\infty_1,\ldots,\infty_s, P_{a_1},\ldots,P_{a_N}\},\mathbb C)$ will give different values of the integral $\kappa_j$ \eqref{kappa} depending on the positions of the points of the surface satisfying $\zeta(P)=z_0$ relative to the contour $\gamma_j.$ Here we assume that initial value for $\Phi$ is given at $z=z_0$. Let us illustrate this with a simple example. 

\begin{example}
\label{example_residues}
Let all contours $\gamma_j$ for $j=1, \dots, p-1$ be equal to the same trivial element $\gamma$ of $H_1(X_a\setminus\{\infty_1,\ldots,\infty_s, P_{a_1},\ldots,P_{a_N}\},\mathbb C)$. In this case, the matrices $B^{(i)}$ given by Theorem \ref{thmDGS} are diagonal: 
\begin{equation}
\label{Bi_example}
B^{(i)}={\rm diag}\left(\theta_i, \theta_i-\frac{n}{m}, \theta_i-\frac{2n}{m}, \dots, \theta_i-\frac{(p-1)n}{m}\right).
\end{equation}
Consider now the points $P\in X_a$ for which $\zeta(P)=z$ for some $z\in\mathbb C\setminus\{a_1, \dots, a_N\}$. There are $m$ such points and we denote them by $P_z^{(1)}, \dots, P_z^{(m)}.$ The integrals $\kappa_j$ \eqref{kappa} computed with $\gamma_j=\gamma$ depend on how many of the points $P_z^{(k)}$ belong to the interior of $\gamma$, that is to the domain delimited by the contour $\gamma$ homologous to zero on $X_a$.

Let us introduce the shorthand notation
\begin{equation}
\label{Pi}
\Pi=\prod_{i=1}^N(z-a_i)\qquad\mbox{and}\qquad\Pi^\theta=\prod_{i=1}^N(z-a_i)^{\theta_i}
\end{equation}
for $\theta:=(\theta_1, \dots, \theta_N)\in\mathbb C^N$ and consider two cases: 
\begin{itemize}
\item[Case 1: ] The contour $\gamma$ does not enclose any of the points $P_z^{(k)}$. In other words, $\gamma$ is a trivial element of $H_1(X_a\setminus(\{P_{a_1},\ldots,P_{a_N}, \infty_1, \dots, \infty_s\}\cup \{P\in X_a\,|\, \zeta(P)=z\}),{\mathbb C})$. In this case each of the $\kappa_j$ with $\gamma_j=\gamma$ vanishes and the solution $\Phi(z)$ given by Theorem \ref{cjtpiche} is diagonal:  $M$ \eqref{M} is the identity matrix and $\Phi(z)=D:$
\begin{equation}
\label{Phidiag}
\Phi(z)=\mathrm{diag}\left(\Pi^{\theta}, \Pi^{\theta-\frac{n}{m}}, \ldots, \Pi^{\theta-\frac{(p-1)n}{m}}\right),
\end{equation}
where we use $\theta-\frac{n}{m}:=(\theta_1-\frac{n}{m}, \dots, \theta_N-\frac{n}{m}).$
\item[Case 2:] The contour $\gamma$ contains $t< m$ points $P_z^{(1)}, \dots, P_z^{(t)}$ in its interior. In this case, we have for $\kappa_j$
\begin{equation}
\label{kappa_example}
\kappa_j := -\frac{m}{jn}\oint_{\gamma}\frac{w^{jn}}{\zeta - z}\mathrm{d}\zeta=-\frac{m}{jn}2\pi\i\sum_{r=1}^t\underset{P_z^{(r)}}{\rm res}\frac{w^{jn}}{\zeta - z}\mathrm{d}\zeta=-\frac{m}{jn}2\pi\i\sum_{r=1}^t \epsilon_r\Pi^{\frac{jn}{m}},
\end{equation}
where $\epsilon_r$ are different $m$th roots of unity and $\Pi^{\frac{jn}{m}}$ is some fixed holomorphic branch of this multivalued function near $z$. Let us, for simplicity of notation, denote $\alpha_j:=-\frac{m}{jn}2\pi\i\sum_{r=1}^t \epsilon_r$ so that $\kappa_j=\alpha_j\Pi^{\frac{jn}{m}}.$ Now, for the entries above the diagonal of the matrix $M$ given by \eqref{M}, we have: 
\begin{equation*}
M_{kl}=\sum_{q \, \vdash \, (l - k)}\,\prod_{j=1}^{l-k}\frac{\left(\alpha_j\Pi^{\frac{jn}{m}}\right)^{\sigma_q(j)}}{\sigma_q(j)!} = \hat \alpha_{l-k}\Pi^{\frac{(l-k)n}{m}}
\end{equation*}
with the new constants $\hat \alpha_j.$ This gives the following solution to the Fuchsian system \eqref{fuchsyst1} with coefficients $\eqref{Bi_example}$:
\begin{equation}
\label{Phiresidues}
   \Phi(z) = \begin{pmatrix}
        \Pi^\theta & \hat\alpha_1\Pi^\theta & \hat\alpha_2\Pi^\theta & \dots & \hat\alpha_{p-1}\Pi^\theta \\
        0 & \Pi^{\theta-\frac{n}{m}} & \hat\alpha_1\Pi^{\theta-\frac{n}{m}} & \dots &\hat\alpha_{p-2}\Pi^{\theta-\frac{n}{m}}  \\
        0 & 0 &  \Pi^{\theta-\frac{2n}{m}} &\dots & \hat\alpha_{p-3}\Pi^{\theta-\frac{2n}{m}}  \\
        \vdots & \vdots & \vdots & \ddots & \vdots \\
        0 & 0 & 0 & 0 &  \Pi^{\theta-\frac{(p-1)n}{m}} & 
    \end{pmatrix},
\end{equation}
or $\Phi_{kl}=\hat \alpha_{l-k}\Pi^{\theta-\frac{(k-1)n}{m}}$ for $l\geqslant k$ and $\Phi_{kl}=0$ otherwise.
\end{itemize}
We thus obtain two solutions \eqref{Phidiag} and \eqref{Phiresidues} for the same system \eqref{fuchsyst1} with coefficients \eqref{Bi_example}. These two fundamental matrices are related, as expected, by the multiplication from the right with a constant upper triangular unipotent matrix $C=(C_{ij})$ for which $C_{ij} = \hat\alpha_{j-i}$ for $j>i$. 
\end{example}

\bigskip

To prepare for the proof of Theorem \ref{cjtpiche}, we first establish the following lemma.
\begin{lemma}\label{prop1}
    Let $M$ and $D$ be the upper-triangular $p \times p$ matrices defined in Theorem \ref{cjtpiche} and $\kappa'_r$ be the derivative \eqref{derivativekappa} of $\kappa_r$ with respect to $z$. Then,
    \begin{enumerate}
        \item $\frac{\mathrm{d}}{\mathrm{d}z}D = D_z = \mathrm{diag}\left( \sum_{i=1}^N\frac{\beta_1^{(i)}}{z-a_i}\prod_{j=1}^N(z - a_j)^{\beta_1^{(j)}}, \ldots,  \sum_{i=1}^N\frac{\beta_p^{(i)}}{z-a_i}\prod_{j=1}^N(z-a_j)^{\beta_p^{(j)}}\right)$,
        \item $\frac{\mathrm{d}}{\mathrm{d}z}M = M_z$  for which the $(kl)$-entry is
        \begin{equation*}
            (M_z)_{kl} = \begin{dcases*}\sum_{q \, \vdash \, (l - k)}\sum_{r=1}^{l-k}\left(\frac{\kappa_r'}{\kappa_r}\sigma_q(r)\prod_{j=1}^{l-k}\frac{\kappa_j^{\sigma_q(j)}}{\sigma_q(j)!}\right) \quad &\text{when } $l>k$,\\  
    \quad 0 \quad &\text{when } $l \leq k$,
    \end{dcases*}
        \end{equation*}
        \item $D^{-1} = \mathrm{diag}\left(\prod_{j=1}^N(z-a_j)^{-\beta_1^{(j)}}, \ldots, \prod_{j=1}^N(z-a_j)^{-\beta_p^{(j)}}\right)$, and
        \item the inverse of the matrix $M$ has the following entries
        \begin{equation}
        \label{M-1}
            (M^{-1})_{kl} = \begin{dcases*}\sum_{q \, \vdash \, (l - k)}\,\prod_{j=1}^{l-k}\frac{(-\kappa_j)^{\sigma_q(j)}}{\sigma_q(j)!} \quad &\text{when } $l>k$,\\ 
    \quad 1 \quad &\text{when } $l= k$,\\ 
    \quad 0 \quad &\text{when } $l < k$.
    \end{dcases*}
        \end{equation}
    \end{enumerate}
\end{lemma}
\begin{proof} Items (1) - (3) of the lemma being obtained by  straightforward observation, let us prove item (4).
      We want to show that
        \begin{equation}
        \label{Minverse}
            (MM^{-1})_{kl} =  \delta_{lk}
        \end{equation}
	with $M^{-1}$ given in the statement of the proposition. Since both $M$ \eqref{M} and $M^{-1}$ \eqref{M-1} are upper triangular with 1's on the main diagonal, the same holds for their product: $(MM^{-1})_{kl}  = \delta_{lk}$ for $l\leqslant k.$ Let us consider the case when $l > k$. Using that $M_{kn}=(M^{-1})_{kn}=0$ for $k>n$,  we have
        \begin{align*}
            (MM^{-1})_{kl} &= \sum_{n=1}^{p}M_{kn}(M^{-1})_{nl} 
            = \sum_{n=k}^{l}M_{kn}(M^{-1})_{nl}\\
            &= \sum_{n = k}^{l}\left[\left(\sum_{r_1 \, \vdash \, (n - k)}\,\prod_{j_1=1}^{n-k}\frac{\kappa_{j_1}^{\sigma_{r_1}(j_1)}}{\sigma_{r_1}(j_1)!}\right)\left(\sum_{r_2 \, \vdash \, (l - n)}\,\prod_{j_2=1}^{l-n}\frac{(-\kappa_{j_2})^{\sigma_{r_2}(j_2)}}{\sigma_{r_2}(j_2)!}\right)\right].
        \end{align*}
Let us show that this expression is zero. Let $T$ be a strictly positive integer and $K, L\geqslant 0.$ Using that $\sigma_q(j)=0$ for $q\vdash T$ and $j>T$, note that
\begin{multline*}
\sum_{K+L = T } \left( \sum_{r_1 \vdash K } \prod_{j=1}^K \frac{\kappa_j^{\sigma_{r_1}(j)}}{\sigma_{r_1}(j)!} \right)\left( \sum_{r_2 \vdash L } \prod_{j=1}^L \frac{(-\kappa_j)^{\sigma_{r_2}(j)}}{\sigma_{r_2}(j)!} \right)
= \sum_{K+L = T } \sum_{r_1 \vdash K } \sum_{r_2 \vdash L } \prod_{j=1}^T \frac{\kappa_j^{\sigma_{r_1}(j)}(-\kappa_j)^{\sigma_{r_2}(j)}}{\sigma_{r_1}(j)!\sigma_{r_2}(j)!} \\
= \sum_{q \vdash T }  \sum_{\substack{(r_1, r_2)\\ r_1 + r_2 = q }} \prod_{j=1}^T \frac{\kappa_j^{\sigma_{r_1}(j)}(-\kappa_j)^{\sigma_{r_2}(j)}}{\sigma_{r_1}(j)!\sigma_{r_2}(j)!}.
\end{multline*}
Here $r_1 + r_2 $ stands for the sum of two partitions combining the parts of $r_1$ and $r_2$, such that $\sigma_{r_1}(j)+\sigma_{r_2}(j)=\sigma_{r_1 + r_2}(j) $ for any $j$ and for two partitions $q_1$ and $q_2$ the pairs $(q_1,q_2)$ and $(q_2, q_1)$ are considered different. 
Since the numbers $\sigma_{r_1}(1), \dots \sigma_{r_1}(T)$ and $\sigma_{r_2}(1), \dots \sigma_{r_2}(T)$ define the partition $r_1$ and $r_2$ completely, we have for every fixed partition $q\vdash T$
\begin{equation}
\label{prod-sum}
\prod_{j=1}^T \,\sum_{ l_j+m_j = \sigma_q(j)}  \frac{\kappa_j^{l_j}(-\kappa_j)^{m_j}}{l_j!m_j!}
=   \sum_{\substack{(r_1, r_2)\\ r_1 + r_2 = q }} \prod_{j=1}^T \frac{\kappa_j^{\sigma_{r_1}(j)}(-\kappa_j)^{\sigma_{r_2}(j)}}{\sigma_{r_1}(j)!\sigma_{r_2}(j)!}.
\end{equation}
Thus, continuing our calculation, we obtain
\begin{multline}
\label{temp}
\sum_{K+L = T } \left( \sum_{r_1 \vdash K } \prod_{j=1}^K \frac{\kappa_j^{\sigma_{r_1}(j)}}{\sigma_{r_1}(j)!} \right)\left( \sum_{r_2 \vdash L } \prod_{j=1}^L \frac{(-\kappa_j)^{\sigma_{r_2}(j)}}{\sigma_{r_2}(j)!} \right)
\\
= \sum_{q \vdash T } \prod_{j=1}^T \frac{1}{\sigma_q(j)!}\sum_{ l_j+m_j = \sigma_q(j)}  \frac{\sigma_q(j)!}{l_j!m_j!}\kappa_j^{l_j}(-\kappa_j)^{m_j} 
= \sum_{q \vdash T } \prod_{j=1}^T \frac{1}{\sigma_q(j)!}(\kappa_j-\kappa_j)^{\sigma_q(j)}=0,
\end{multline}
where we set $(\kappa_j-\kappa_j)^0=1$ for notational convenience. 
For $T=l-k$, this implies \eqref{Minverse} with  $l>k$ and thus proves that \eqref{M-1} defines the inverse of the matrix $M$ \eqref{M}.
\end{proof}
We now have all the necessary tools to prove Theorem \ref{cjtpiche}.
\\

\begin{customproof}[Theorem \ref{cjtpiche}]
We want to show that $\Phi = MD$
where both $M$ and $D$ are defined in Theorem \ref{cjtpiche}, is a fundamental solution of the Fuchsian system \eqref{fuchsyst1}, i.e.
\begin{equation*}
\Phi_z \Phi^{-1} = \sum_{i=1}^N\frac{1}{z-a_i}B^{(i)},
\end{equation*}
where $ \Phi_z = \frac{\mathrm{d}}{\mathrm{d}z}\Phi.$ Differentiating \eqref{productMD}, we have
\begin{equation}
\label{derivative}
    \Phi_z \Phi^{-1} = (M_z D + MD_z)D^{-1}M^{-1} = M_z M^{-1} + MD_z D^{-1}M^{-1}.
\end{equation}
The matrix $M_z M^{-1}$ is upper triangular as a product of upper triangular matrices. Its diagonal entries are zeros since this is the case for $M_z$, see Lemma \ref{prop1}. 
Let us consider the entries of $M_z M^{-1}$ lying above the main diagonal.   Fix $k<l$ and set $j=l-k$.  
By Lemma \ref{prop1}, for $r\geqslant 1$, we have
\begin{equation*}
 (M_z)_{k,k+r}
 = \sum_{q_1\vdash r} \sum_{t=1}^r 
 \frac{\kappa_t'}{\kappa_t}\;\sigma_{q_1}(t) \prod_{i=1}^r\frac{\kappa_i^{\sigma_{q_1}(i)}}{\sigma_{q_1}(i)!}
\end{equation*}
and for $1\leqslant r < j$
\begin{equation}
\label{M-1r}
(M^{-1})_{k+r,k+j}
  = \sum_{q_2\vdash (j-r)}
    \prod_{i=1}^{j-r}\frac{(-\kappa_i)^{\sigma_{q_2}(i)}}{\sigma_{q_2}(i)!},
\end{equation}
hence
\begin{equation}
\label{temp-comb}
  (M_zM^{-1})_{kl}
  = \sum_{r=0}^j \;
    \sum_{q_1\vdash r} \;
    \sum_{q_2\vdash (j-r)} \;
    \sum_{t=1}^r
    \frac{\kappa_t'}{\kappa_t}\,\sigma_{q_1}(t)
    \left(\prod_{i=1}^r\frac{\kappa_i^{\sigma_{q_1}(i)}}{\sigma_{q_1}(i)!} \right)
    \;\prod_{h=1}^{j-r}\frac{(-\kappa_h)^{\sigma_{q_2}(h)}}{\sigma_{q_2}(h)!}.
\end{equation}
In this proof, we assume that there exists one partition of zero. 
Grouping the partitions $q_1$ and $q_2$ into a sum, we obtain the sum over all partitions $q:=q_1+ q_2$ of $j$ with $\sigma_{q_1}(i)+\sigma_{q_2}(i)=\sigma_{q}(i) $.  This gives for $j=l-k\geqslant 1$
\begin{equation*}
  (M_zM^{-1})_{kl}
  = \sum_{q\vdash j} \;\sum_{\substack{(q_1, q_2)\\ q_1 + q_2 = q\vdash j }}\;
   \sum_{t=1}^j
   \frac{\kappa_t'}{\kappa_t}\,\sigma_{q_1}(t) \; \prod_{i=1}^j
   \binom{\sigma_q(i)}{\sigma_{q_1}(i)}(-1)^{\sigma_{q_2}(i)}\frac{\kappa_i^{\sigma_q(i)}}{\sigma_q(i)!} 
\end{equation*}
or, equivalently,
\begin{equation}
\label{MzM-temp}
  (M_zM^{-1})_{kl}
  = \sum_{q\vdash j} \;  \left( \prod_{i=1}^j\frac{\left(\kappa_i\right)^{\sigma_q(i)}}{\sigma_q(i)!}\right)   \sum_{\substack{(q_1, q_2)\\ q_1 + q_2 = q\vdash j }}
   \sum_{t=1}^j
   \frac{\kappa_t'}{\kappa_t}\,\sigma_{q_1}(t) \; \prod_{i=1}^j (-1)^{\sigma_{q_2}(i)}
   \binom{\sigma_q(i)}{\sigma_{q_1}(i)} .
\end{equation}
Let us now consider the contribution, denoted by $C(q,t),$ to \eqref{MzM-temp} of a fixed partition $q\vdash j$ and of a fixed value of $t$. Introduce the vector $\sigma_q=(\sigma_q(1), \dots, \sigma_q(j))$, which defines the partition $q$ uniquely.  We can rewrite the contribution  of  $q$ and $t$ in \eqref{MzM-temp} as a sum over all decompositions of $\sigma_q$ into the sum of two vectors. Then, similarly to \eqref{prod-sum}, this becomes the product of sums (we write $v_k(i)$ for the $i$th component of the vector $v_k$):
\begin{align}
\label{Cqt}
\nonumber
  C(q,t)=\sum_{\substack{(q_1, q_2)\\ q_1 + q_2 = q\vdash j }}
\sigma_{q_1}(t) \; \prod_{i=1}^j (-1)^{\sigma_{q_2}(i)}
   \binom{\sigma_q(i)}{\sigma_{q_1}(i)}
    = \sum_{v_1+v_2=\sigma_q}
v_{1}(t) \; \prod_{i=1}^j (-1)^{v_{2}(i)}
   \binom{\sigma_q(i)}{v_{1}(i)}
   \\
   =\left( \sum_{x_t=0}^{\sigma_q(t)} x_t \binom{\sigma_q(t)}{x_t} (-1)^{\sigma_q(t)-x_t} \right) \prod_{\substack{i=1\\i \ne t}}^j \left( \sum_{x_i=0}^{\sigma_q(i)} \binom{\sigma_q(i)}{x_i} (-1)^{\sigma_q(i)-x_i} \right).
\end{align}
If $\sigma_q(i)\neq 0$, the sum over $x_i$ for $i\neq t$ equals $(1+x)^{\sigma_q(i)}$ evaluated at $x=-1$ and thus vanishes. If $\sigma_q(i)=0$, the sum over $x_i$ is equal to $1$. The sum over $x_t$, if $\sigma_q(t)>0$, can be rewritten as the derivative of $(1+x)^{\sigma_q(t)}$ evaluated at $-1$, therefore it is equal to $1$ if $\sigma_q(t)=1$ and $0$ otherwise. Thus, $C(q,t)$ is non-zero if and only if $q=j$ and $t=j$, that is if $q$ is the partition of $j$ containing one part only, in which case $C(q,t)=1$.
Using this result in \eqref{MzM-temp}, and noting that the only splitting of $j$ into the sum of two partitions that gives a non-zero contribution in \eqref{MzM-temp} is $q_1=j, \; q_2=0$,  we obtain
\begin{equation*}
  (M_zM^{-1})_{kl} =
  \begin{cases}
    0,&l\le k,\\
    \kappa_{\,l-k}',&l>k.
  \end{cases}
\end{equation*}
Let us now consider the term $MD_zD^{-1}M^{-1}$ in \eqref{derivative}. Recalling that  $\beta_h^{(i)}=\beta_k^{(i)}-\tfrac{(h-k)n}m$ and using Lemma \ref{prop1} for $D_z$, one finds that for $l\geqslant k$ 
\begin{align*}
  (MD_zD^{-1}M^{-1})_{kl}
  &= \sum_{i=1}^N\sum_{h=k}^l M_{kh}\,\frac{\beta_h^{(i)}}{\;z-a_i\!}\,(M^{-1})_{hl}
    \\
  &= \sum_{i=1}^N \frac{\beta_k^{(i)}}{z-a_i}\!\sum_{h=k}^lM_{kh}(M^{-1})_{hl}
     \;-\;\frac nm\sum_{i=1}^N\frac1{z-a_i}\sum_{h=k}^l(h-k)M_{kh}(M^{-1})_{hl}
     \\
  &= \sum_{i=1}^N \frac{\beta_k^{(i)}}{z-a_i}\delta_{kl}
     \;-\;\frac nm\sum_{i=1}^N\frac1{z-a_i}\sum_{r=0}^{l-k}rM_{k,k+r}(M^{-1})_{k+r,l}
    .     
\end{align*}
Comparing with our above derivation, we see that the second term  can be rewritten using partitions similarly to \eqref{MzM-temp}, where $j=l-k$. More precisely, the factor of $r$ can be represented as the sum of all the parts of the partition $q_1$, that is $r=\sum_{t=1}^j t\,\sigma_{q_1}(t)$. Then the sum over $r$ in the second term has the form of \eqref{temp-comb} with $\frac{\kappa_t'}{\kappa_t}$ replaced by $t$. In this way, using \eqref{M} for $M$ as well as formulas from \eqref{M-1r} to \eqref{Cqt}, we obtain
\begin{align}
\label{rtemp}
\sum_{r=0}^{l-k}rM_{k,k+r}(M^{-1})_{k+r,l}&=\sum_{q\vdash j} \;  \left( \prod_{i=1}^j\frac{\left(\kappa_i\right)^{\sigma_q(i)}}{\sigma_q(i)!}\right) \!\!\!  \sum_{\substack{(q_1, q_2)\\ q_1 + q_2 = q\vdash j }}\!\!\! \sum_{t=1}^j t\,\sigma_{q_1}(t)
    \prod_{i=1}^j (-1)^{\sigma_{q_2}(i)}
   \binom{\sigma_q(i)}{\sigma_{q_1}(i)}
    \\ \nonumber
   & =
   \sum_{q\vdash j}   \left( \prod_{i=1}^j\frac{\left(\kappa_i\right)^{\sigma_q(i)}}{\sigma_q(i)!}\right) 
   \sum_{t=1}^j t\,C(q,t).
   \end{align}
Given that $C(q,t)$ is non-zero only in the case of one term partition $q=j=l-k$ and $t=j,$ in which case $C(j,j)=1$, we obtain that  quantity \eqref{rtemp} is equal to $j\kappa_j$ and thus
\begin{equation*}
  (MD_zD^{-1}M^{-1})_{kl}=
  \begin{cases}
    0,&l<k,\\
\displaystyle\sum_{i=1}^N\frac{\beta_k^{(i)}}{z-a_i},&l=k,\\
    \displaystyle-\,\frac{(l-k)n}m\,\kappa_{\,l-k}\sum_{i=1}^N\frac1{z-a_i},&l>k.
  \end{cases}
\end{equation*}
Summing the two contributions into \eqref{derivative}, we have
\begin{equation*}
  (\Phi_z\Phi^{-1})_{kl}
  =\begin{cases}
    0,&l<k,\\[3pt]
    \displaystyle\sum_{i=1}^N\frac{\beta_k^{(i)}}{z-a_i},&l=k,\\[4pt]
    \kappa_{l-k}'
    -\displaystyle\frac{({l-k})n}m\,\kappa_{l-k}\sum_{i=1}^N\frac1{z-a_i},&l>k.
  \end{cases}
\end{equation*}
 Recalling that from Lemma \ref{propcalculsimp}, for each $l-k\ge1$, 
\begin{equation*}
\;\kappa_{l-k}' - \frac{({l-k})n}{m}\kappa_{l-k}\sum_{i=1}^N\frac1{z-a_i}
    = \sum_{i=1}^N\frac{B^{(i)}_{k,l}}{z-a_i},
\end{equation*}
we obtain
\begin{equation*}
\Phi_z\Phi^{-1}=\sum_{i=1}^N\frac{B^{(i)}}{z-a_i}.
\end{equation*}
\end{customproof}

\section{Monodromy matrices of solutions}
\label{sect_monodromies}

Let us introduce a basis $\{\mathcal A_1, \dots, \mathcal A_g; \mathcal B_1, \dots, \mathcal B_g\}$ in $H_1(X_a,\mathbb C)$ not passing trough ramification points of the covering $\zeta$ with the intersection indices being $\mathcal A_i\circ \mathcal A_k=\mathcal B_i\circ \mathcal B_k=0$ and $\mathcal A_i\circ \mathcal B_k=1$; such a basis is called {\it canonical}.
The integration contours $\gamma_j$ from Theorem \ref{cjtpiche} are linear combinations of the cycles $\mathcal A_k, \mathcal B_k$ for $k=1,\dots, g$, small loops $\gamma_{a_i}$ encircling the points $P_{a_i}$ for $i=1, \dots,N$, small loops $\gamma_{\infty_\alpha}$ encircling the points $\infty_\alpha$ for $\alpha=1, \dots, s$ as well as small loops $\eta_t$ for $t=1, \dots, m$ each encircling counterclockwise a point $P_z^{(t)}$ for which $\zeta(P)=z.$ Assume that the only pairs of contours that intersect each other are $\mathcal A_k$ and  $\mathcal B_k$ for $k=1, \dots, g.$

Let us identify the $z$- and $\zeta$-spheres and choose generators $\rho_1, \dots, \rho_N$ of the fundamental group $\pi_1(\mathbb C\setminus\{a_1, \dots, a_N\}, z_0)$ in a way so  that $\rho_j$ goes counterclockwise once around the point $z=a_j$  and does not encircle any other branch point  of the covering $\zeta.$

We may analytically continue the fundamental solution $\Phi(z_0)$ from Theorem \ref{cjtpiche} to the Fuchsian system \eqref{fuchsyst1} along the generators $\rho_i$ to obtain another fundamental solution $\Phi_{\rho_i}(z_0)$ of the same system. The system being linear, the two solutions are related by $\Phi_{\rho_i}(z_0)=\Phi(z_0)\mathcal M_i$ for some (monodromy) matrix $\mathcal M_i$. Given that the coefficients $B^{(i)}$ of \eqref{fuchsyst1} solve the Schlesinger system, there exists a fundamental solution to the corresponding Fuchsian system for which the monodromy matrices are independent of small variations of $a_1, \dots, a_N.$

In this section, we calculate the monodromy matrices for fundamental solutions from Theorem \ref{cjtpiche} and show that these solutions are isomonodromic. 

Let us first discuss how the process of analytical continuation of solutions $\Phi(z_0)$ from Theorem \ref{cjtpiche} along the loops $\rho_i$ affects the contours of integration. 
Recall that as $z$ varies, the integration contours deform continuously to avoid the points $P_z^{(t)}$, $t=1, \dots, m,$ for which $\zeta(P)=z$.

The contours $\gamma_{a_i}$ and $\gamma_{\infty_\alpha}$ can be made arbitrarily small and thus are not affected by the analytical continuation of $\Phi(z_0)$ along the loops $\rho_i$. The contours $\eta_t$ for $t=1, \dots, m$ follow the points $P_z^{(t)}$ and thus get permuted during the analytical continuation in question. As for the contours $\mathcal A_k$ and  $\mathcal B_k$, they may remain invariant or may transform by adding several cycles from the set  $\eta_1, \dots, \eta_m$, see Example \ref{ex_contours}. Note that an integration contour cannot transform by adding cycles $\mathcal A_k$ and  $\mathcal B_k$ as $z$ goes around a loop $\rho_i$ since the points $P_z^{(t)}$ during such a transformation do not follow closed cycles. 
\begin{example}
\label{ex_contours}
Consider the compact Riemann surface $X_a$ corresponding to the curve of equation \eqref{curves} with $m=2$. Then the ramified covering $\zeta:X_a \to \mathbb CP^1$ is two-fold. Suppose a cycle $\mathcal A_k$ goes around two ramification points $P_{a_1}$ and $P_{a_2}$ as in the left part of Figure \ref{fig_ex} (note that the figure only shows a part of the covering close to the cycle in question). 
\begin{figure}[h]
\includegraphics[height=4cm, width=6cm]{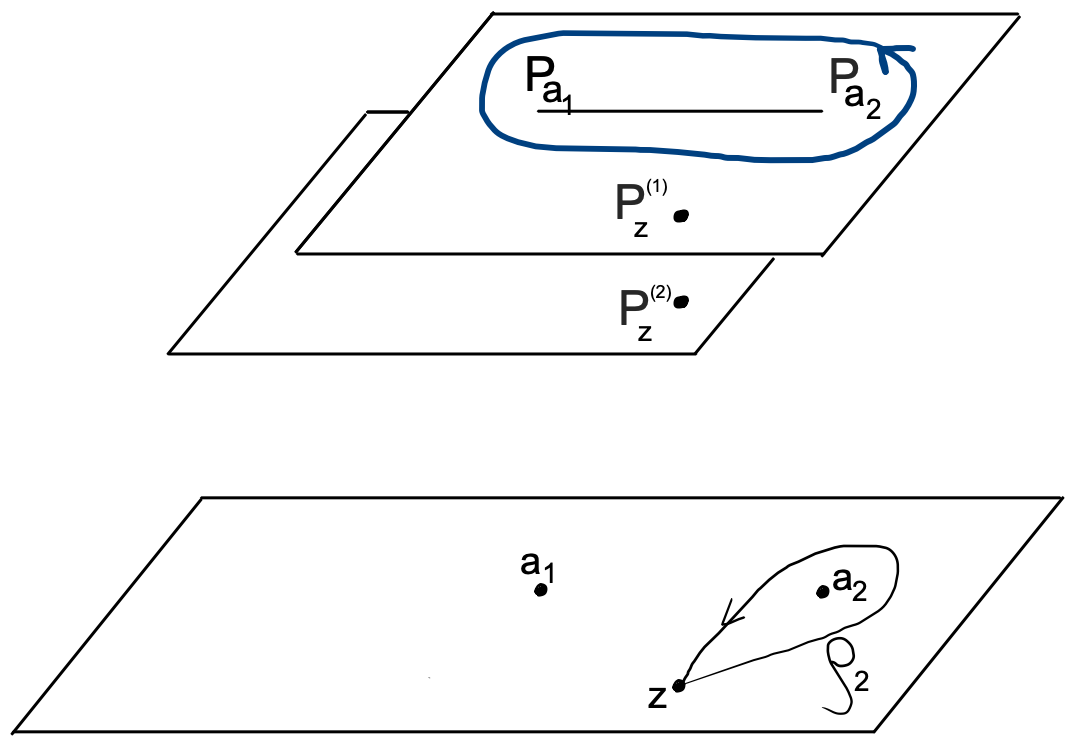}
\hspace{2cm}
\includegraphics[height=4cm, width=6cm]{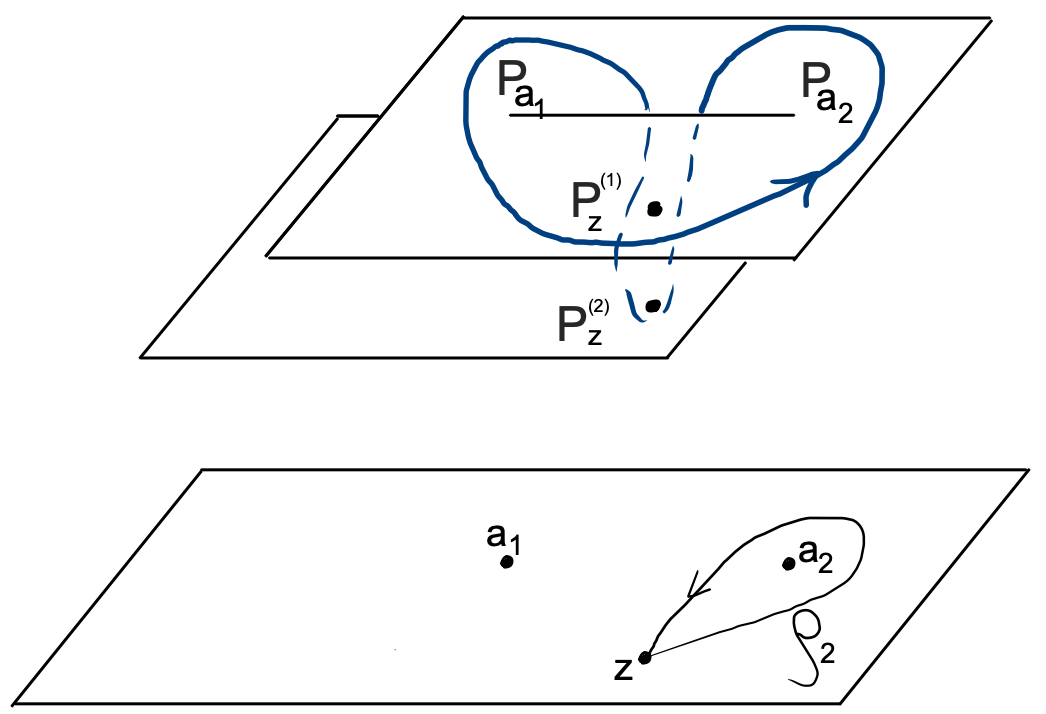}
\caption{ Transformation of a contour as $z$ follows the loop $\rho_2$}
\label{fig_ex}
\end{figure}
In the case of a two-sheeted covering, we may assume 
there is a branch cut between the two ramification points, that is a line by crossing which a path goes from one sheet to the other. The points $P_z^{(1)}$ and $P_z^{(2)}$ belong to two different sheets and project to $\zeta=z.$ The right part of Figure \ref{fig_ex} shows the transformation of $\mathcal A_k$ as $z$ goes once around the loop $\rho_{a_2}$. We see that the transformed contour is equal to $\mathcal A_k+\eta_1-\eta_2$
and thus  the effect on the integral in \eqref{kappa} along $\gamma_r:=\mathcal A_k$ is non-trivial. 
\end{example}

We thus can decompose an integration contour $\gamma_r$ from Theorem \ref{cjtpiche} into a sum of contours of three types: $\gamma_r=\gamma_{r1}+\gamma_{r2}+\gamma_{r3}$, where $\gamma_{r1}$ is a linear combination of contours that do  not transform as $z$ goes around a loop $\rho_i$, the contour $\gamma_{r2}$ is a linear combination of contours $\eta_1, \dots, \eta_m$, and $\gamma_{r3}$ is a linear combination of contours which transform by adding a combination of contours $\eta_t$ (note that linear combinations $\gamma_{rk}$ may be trivial).

This decomposition of $\gamma_r$ induces a decomposition of the function $\kappa_r$ defined by \eqref{kappa}: 
\begin{equation*}
\kappa_r=-\frac{m}{rn} \left( \oint_{\gamma_{r1}}\frac{w^{rn}}{\zeta - z}\mathrm{d}\zeta+ \oint_{\gamma_{r2}}\frac{w^{rn}}{\zeta - z}\mathrm{d}\zeta+ \oint_{\gamma_{r3}}\frac{w^{rn}}{\zeta - z}\mathrm{d}\zeta\right).
\end{equation*}
Note that the integral over $\gamma_{r2}$ is a linear combination of residues at the points $P_z^{(t)}$ calculated in \eqref{kappa_example}. The integral over $\gamma_{r2}$ is thus proportional to $\Pi^{\frac{rn}{m}}.$ 
Under the analytical continuation along $\rho_i$ the integral over $\gamma_{r3}$ transforms by adding a linear combination of residues at the points $P_z^{(t)}$, thus by adding a function proportional to $\Pi^{\frac{rn}{m}}.$ Altogether, we have the following transformation of $\kappa_r$ as $z$ goes around the loop $\rho_i:$
\begin{equation*}
\kappa_r \overset{\rho_i}{\longmapsto } -\frac{m}{rn} \left( \oint_{\gamma_{r1}}\frac{w^{rn}}{\zeta - z}\mathrm{d}\zeta+ \e^{\frac{2\pi\i rn}{m}}\oint_{\gamma_{r2}}\frac{w^{rn}}{\zeta - z}\mathrm{d}\zeta+ \oint_{\gamma_{r3}}\frac{w^{rn}}{\zeta - z}\mathrm{d}\zeta + c_i\Pi^{\frac{rn}{m}}\right)
\end{equation*}
for some constant $c_i\in\mathbb C$, or, equivalently, 
\begin{equation}
\label{kappa-transformation}
\kappa_r \overset{\rho_i}{\longmapsto } \kappa_r  + c_{ri}\Pi^{\frac{rn}{m}},
\end{equation}
where $c_{ri}=-\frac{m}{rn}(c_i+\alpha\e^{\frac{2\pi\i rn}{m}}-\alpha)$ where $\oint_{\gamma_{r2}}\frac{w^{rn}}{\zeta - z}\mathrm{d}\zeta=\alpha\Pi^{\frac{rn}{m}}$ for some $\alpha\in\mathbb C$. 

\begin{thm}
\label{thm_monodromy} 
Let $X_a, a\in D$  be compact Riemann surfaces as in Theorem \ref{thmDGS}. Let $z_0\in\mathbb C\setminus\{a_1, \dots, a_N\}$ and $\gamma_1,\dots, \gamma_{p-1}$ be  closed contours on $X_a$   from $H_1(X_a\setminus(\{\infty_1,\ldots,\infty_s\}\cup\{P_{a_1},\ldots,P_{a_N}\}\cup \{P\in X_a\,|\, \zeta(P)=z_0\}),{\mathbb C})$. Let $\Phi(z_0)$ be the fundamental matrix from Theorem \ref{cjtpiche} constructed using the contours $\gamma_1,\dots, \gamma_{p-1}$ and exponents $\beta_k^{(i)}= \theta_i - \frac{(k-1)n}{m},$ with  $k \in \{1, \ldots, p\}$ and $i \in \{1, \ldots, N\}$. Let $\rho_1, \dots, \rho_N$ be generators of $\pi_1(\mathbb C\setminus\{a_1, \dots, a_N\}, z_0)$ where $\rho_i$ goes anticlockwise once around the point $z=a_i$  and does not encircle any other $a_k$ with $k\neq i.$ Denote by $\Phi_{\rho_i}(z_0)$ the fundamental matrix of system \eqref{fuchsyst1} obtained by analytical continuation of  $\Phi(z_0)$ along   the loop $\rho_i$. 
Let the constants $c_{ri}\in\mathbb C$ be defined by the transformation of the functions $\kappa_r$ \eqref{kappa} under the analytic continuation  along the loop $\rho_i$ as in \eqref{kappa-transformation}. 
Then $\Phi_{\rho_i}(z_0)=\Phi(z_0)\mathcal M_i$ with the monodromy matrix $\mathcal M_i$ given by 
\begin{equation}
\label{Mon-product}
 \mathcal M_i = \hat {\mathcal M}_iD_i,
\end{equation}
where
\begin{equation}
\label{MonD}
    D_i = \mathrm{diag}\left( \e^{2\pi\i\beta_1^{(i)}}, \ldots, \e^{2\pi\i\beta_p^{(i)}}\right),
\end{equation}
and the $(kl)$-entry of $\hat {\mathcal M}_i$ is
\begin{equation}
\label{MonM}
    (\hat {\mathcal M}_i)_{kl} = \begin{dcases*}\sum_{q \, \vdash \, (l - k)}\,\prod_{r=1}^{l-k}\frac{c_{ri}^{\sigma_q(r)}}{\sigma_q(r)!} \quad &\text{when } $l>k$,\\ 
    \quad 1 \quad &\text{when } $l= k$,\\ 
    \quad 0 \quad &\text{when } $l < k$,
    \end{dcases*}
\end{equation}
where, as before,  $q \vdash (l-k)$ represents the partitions of the integer $l - k$,  and $\sigma_q(r)$ is the number of repetitions of the part $r$ in the partition $q$.
\end{thm}
{\it Proof.} Using the notation $\Pi$ introduced in \eqref{Pi} and the vectors $\beta_k=(\beta_k^{(1)}, \dots, \beta_k^{(N)})$, we can write the result of the analytical continuation along the loop $\rho_i$ in the $z$-sphere of the fundamental solution $\Phi(z)$ given by \eqref{productMD}-\eqref{M} of Theorem \ref{cjtpiche} as
\begin{equation*}
\Phi_{\rho_i}(z) = M_{\rho_i}\mathrm{diag}\left(\Pi^{\beta_1}\e^{2\pi\i\beta_1^{(i)}}, \ldots, \Pi^{\beta_p}\e^{2\pi\i\beta_p^{(i)}}\right),
\end{equation*}
where $M_{\rho_i}$ is the corresponding analytical continuation of the matrix $M$ \eqref{M}. Given the transformation \eqref{kappa-transformation} of the functions $\kappa_r$, we have for the $(kl)$-entry of $\Phi_{\rho_i}(z)$ with $l>k:$
\begin{equation}
\label{kl-transformation}
(\Phi_{\rho_i})_{kl} =  \Pi^{\beta_l}\e^{2\pi\i\beta_l^{(i)}}  \sum_{q \, \vdash \, (l - k)}\,\prod_{r=1}^{l-k}\frac{\left(\kappa_r+c_{ri}\Pi^{\frac{rn}{m}}\right)^{\sigma_q(r)}}{\sigma_q(r)!}.
\end{equation}
Let us now compute the $(kl)$-entry of the matrix $\Phi\mathcal M_i$ with $\mathcal M_i$ given by \eqref{Mon-product}-\eqref{MonM}:
\begin{align}
\label{kl-mon}
(\Phi{\mathcal M}_i)_{kl} &=\e^{2\pi\i\beta_l^{(i)}}  \sum_{t=k}^l M_{kt} (\hat{\mathcal M}_i )_{tl}\Pi^{\beta_t} 
\\ \nonumber
&=  \e^{2\pi\i\beta_l^{(i)}} \Pi^{\beta_l} \sum_{t=k}^l \left(\sum_{q_1 \, \vdash \, (t - k)}\,\prod_{r_1=1}^{t-k}\frac{\kappa_{r_1}^{\sigma_{q_1}(r_1)}}{\sigma_{q_1}(r_1)!} \right) 
\left(\sum_{q_2 \, \vdash \, (l - t)}\,\prod_{r=1}^{l-t}\frac{\left(c_{ri}\Pi^{\frac{rn}{m}}\right)^{\sigma_{q_2}(r)}}{\sigma_{q_2}(r)!} \right),
\end{align}
where we represented $\Pi^{\beta_t} = \Pi^{\beta_l} \Pi^{\frac{(l-t)n}{m}}.$ By the argument used to obtain \eqref{temp} with $K=t-k$ and $L=l-t$, we conclude that expression \eqref{kl-mon} is equal to \eqref{kl-transformation}, that is $(\Phi\mathcal M_i)_{kl}=(\Phi_{\rho_i})_{kl}$ for $l>k.$ For $l\leqslant k$, the statement of the theorem is straightforward. 
$\Box$

\bigskip

Given that the constants $c_{ri}$ in \eqref{MonM} are independent of $\{a_1, \dots, a_N\}$,  
Theorem \ref{thm_monodromy} implies that monodromy matrices of the fundamental solution given by Theorem \ref{cjtpiche} are independent of small variations of $a_i$, $i=1, \dots, N$, that is the solution $\Phi(z)$ given by \eqref{productMD}-\eqref{M} is isomonodromic. The monodromy matrices \eqref{Mon-product}-\eqref{MonM} are conjugate to ${\rm exp} \{2\pi\i\,B^{(i)}\}$ as expected according to the general theory of non-resonant Fuchsian systems, see Chapter IV, Theorem 4.1 of \cite{CL}. 
Note that our solutions also transform when $z$ goes around the point at infinity, and the corresponding monodromy matrix $\mathcal M_\infty$ is conjugate to ${\rm exp} \{2\pi\i\,B^{(\infty)}\}$ where $B^{(\infty)}=-\sum_{i=1}^N B^{(i)}$ is a nonzero diagonal matrix, see the end of Section \ref{sect_setup}. A loop around $z=\infty$ is given by a product of loops $\rho_i$ in the fundamental group and thus $\mathcal M_\infty^{-1}$ is obtained as the corresponding product of the matrices $\mathcal M_i$.

\section{Some solutions $\Phi=MD$ with rational $M$}
\label{sect_examples}

As mentioned after Theorem \ref{thmDGS}, the contours of integration may be chosen in a way for the solutions of the Schlesinger system from Theorem \ref{thmDGS} to be given by residues of the differentials \eqref{diffmero}. From \cite{DragGontShram}, we know that for $n>0$, the differentials only have poles at the $s$ points at infinity $\infty_1, \dots, \infty_s$. For coprime $m$ and $N$, there is only one pole at infinity and thus the residue there is zero. For $n<0$, all poles are at the finite ramification points $P_{a_1}, \dots, P_{a_N}.$ In these cases we obtain, respectively, polynomial and rational solutions to the Schlesinger system. For the corresponding solutions $\Phi=MD$ to the Fuchsian system, the matrix $M$ is rational. The solution $\Phi$ is algebraic for rational constants $\theta_i$, $i=1, \dots, N,$ from Theorem \ref{cjtpiche}.

\subsection{Positive values of $n$}
\label{sect_rat1}

Following \cite{DragGontShram}, let us denote $N=sN_1\,,$ $m=sm_1\,$ for $s={\rm gcd}(m, N)$ and recall that $n>0$ should be coprime with $m$.  In this case the differential $\Omega_i^{(j)}(a)$ \eqref{diffmero} has $s$ poles, one at each
of the points $\infty_1,\ldots,\infty_s$ at infinity of the Riemann surface $X_a$.  The behaviour of these differentials at $P=(\zeta,w)\sim \infty_\alpha$ for  $\alpha=1, \dots, s$ in terms of the local parameter $\xi:=\xi_{\infty_\alpha}$ given by \eqref{coordinates} is 
\begin{equation}
\label{omegab}
\Omega_i^{(j)}(a)=\frac{w^{jn} d\zeta}{\zeta-a_i}=\frac{\nu_{jn\alpha}\prod_{k=1}^N(1-a_k\xi^{m_1})^{jn/m}}{\xi^{jnN_1+1}(1-a_i\xi^{m_1})}\,d\xi,
\qquad \nu_{jn\alpha}=-m_1\,e^{2\pi\i jn(\alpha-1)/s}.
\end{equation}
We thus have a pole of order $jnN_1+1$ at $P=\infty_\alpha$ for each $\Omega_i^{(j)}$.  As can be seen from \eqref{omegab} the residue at the pole is nonzero only if $j$ is a multiple of $m_1$. 

Note that the constant overall factor in \eqref{omegab} was unimportant in \cite{DragGontShram} but for our purposes in the present paper, we need to keep track of it. To this end, it is convenient introduce the following notation
\begin{equation*}
A_j^{(\alpha)}:=(-1)^{N\frac{jn}{m}}\nu_{jn\alpha}.
\end{equation*}
Let $\gamma_{\infty_\alpha}$ be  a small contour encircling $\infty_\alpha$. Choosing now $\gamma_{l-k}=\frac{c_{l-k}}{2\pi\i A_{l-k}^{(\alpha)}}\gamma_{\infty_\alpha}$ for some $c_{l-k}\in\mathbb C$, we have a polynomial solution to the Schlesinger system given, for some $\theta_i\in\mathbb C$, by (see Theorem 3 of \cite{DragGontShram})
\begin{equation}
\label{polySchlesinger}
B^{(i)}_{kl}=
\begin{dcases*}     
\frac{c_{l-k}}{A_{l-k}^{(\alpha)}}\;\underset{\infty_\alpha}{\rm res}\,\Omega_i^{(l-k)}(a),\quad\text{if }  (l-k)/m_1\in\mathbb Z \;\text{ and } \; l>k \\ 
0,\qquad\qquad\qquad\qquad\;\;\, \text{if }  (l-k)/m_1\notin\mathbb Z \;\text{ or } \; l<k \\
\theta_i - \frac{(k-1)n}{m}, \qquad\quad \;\;\,\text{if } \; l=k.
\end{dcases*} 
\end{equation}
 The nonzero residue from \eqref{polySchlesinger}, for $(l-k)/m_1\in\mathbb Z $, is explicitly given by \cite{DragGontShram}
\begin{equation*}
\underset{\infty_\alpha}{\rm res}\,\Omega_i^{(l-k)}(a)= A_{l-k}^{(\alpha)}
\sum_{k_1+\ldots+k_N+\q=N(l-k)\frac{n}{m} }(-1)^\q{(l-k)\frac{n}{m}\choose k_1}\ldots{(l-k)\frac{n}{m}\choose k_N}a_1^{k_1}\ldots a_N^{k_N}a_i^\q, 
\end{equation*}
with the generalized binomial coefficients 
\begin{equation*}
{\beta\choose j}=\frac{\beta(\beta-1)\cdots(\beta-j+1)}{j!}, \qquad {\beta\choose 0}=1
\end{equation*}
defined for $\beta\in\mathbb R$ and $j\in\mathbb N$.  
Solutions obtained in Theorem \ref{cjtpiche} to the Fuchsian system \eqref{fuchsyst1} with polynomial coefficients \eqref{polySchlesinger} are  functions of $z$ given by the next corollary. These solutions are algebraic for rational values of $\theta_i$.

\begin{corollary}
\label{cor_rat1}
Let  $\theta_i \in \mathbb{C}$ for $i=1, \dots, N$ and $p, n, m, N$ be strictly positive integers such that ${\rm gcd}(n,m)=1$  and $s={\rm gcd}(m, N)>1$.  With $c_j\in\mathbb C$ as in \eqref{polySchlesinger}, let $R_j=R_j(z)$ be the polynomial
\begin{equation}
\label{Rj}
R_j(z)=
\begin{dcases*}
-\frac{m\,c_{j}}{jn}\sum_{k_1+\ldots+k_N+\q=N\frac{jn}{m} }(-1)^\q{{jn}/{m}\choose k_1}\ldots{{jn}/{m}\choose k_N}a_1^{k_1}\ldots a_N^{k_N}z^\q ,\quad\text{if }  j/m_1\in\mathbb Z \\ 
0,\qquad\qquad\qquad\qquad\qquad\qquad\qquad\qquad\qquad\qquad\qquad\qquad\qquad\quad\;\text{if }  j/m_1\notin\mathbb Z .
\end{dcases*}
\end{equation}
Denote $\beta_k^{(i)}=\theta_i - \frac{(k-1)n}{m}$ for $k=1, \dots, p$ and $i=1, \dots, N$.
Then the following upper triangular matrix $\Phi(z)$ solves the Fuchsian system \eqref{fuchsyst1} with coefficients $B^{(i)}$ being the upper triangular $p \times p$ matrices given by \eqref{polySchlesinger}:
\begin{equation*}
    \Phi = MD,
\end{equation*}
where
\begin{equation*}
    D = \mathrm{diag}\left(\prod_{i=1}^N(z-a_i)^{\beta_1^{(i)}}, \ldots, \prod_{i=1}^N(z-a_i)^{\beta_p^{(i)}}\right),
\end{equation*}
and
\begin{equation*}
    M_{kl} = \begin{dcases*}\sum_{q \, \vdash \, (l - k)}\,\prod_{j=1}^{l-k}\frac{R_j^{\sigma_q(j)}}{\sigma_q(j)!} \quad &\text{when } $l>k$,\\ 
    \quad 1 \quad &\text{when } $l= k$,\\ 
    \quad 0 \quad &\text{when } $l < k$,
    \end{dcases*}
\end{equation*}
where $\sigma_q(j)$ is the number of repetitions of the part $j$ in the partition $q$.
\end{corollary}

\begin{proof}
Let $z\in\mathbb C\setminus\{a_1, \dots, a_N\}$. For $n>0$ and the compact Riemann surface $X_a$ corresponding to the superelliptic curve \eqref{curves}, the differentials $\Omega_i^{(j)}$ \eqref{diffmero} have poles at $\infty_1, \dots, \infty_s$ as explained at the beginning of Section \ref{sect_rat1}. Comparing the statements of Corollary \ref{cor_rat1} and  Theorem \ref{cjtpiche}, we only need to prove that $\kappa_j$ from Definition \ref{dfnkappa} coincides with $R_j$ \eqref{Rj} for the choice of the integration contour $\gamma_j = \frac{c_{j}}{2\pi\i A_j^{(\alpha)}}\gamma_{\infty_\alpha}$ with $\gamma_{\infty_\alpha}$ being a small contour encircling $\infty_\alpha\in X_a$ and not containing in its interior any point $P\in X_a$ for which $\zeta(P)=z$. With this choice and $m_1, N_1$ as above, we have
\begin{equation*}
\kappa_j = -\frac{m\,c_j}{jnA_j^{(\alpha)}}\,\underset{\infty_\alpha}{\rm res}\,\frac{w^{jn}\mathrm{d}\zeta}{\zeta - z} = -\frac{m\,c_j}{jnA_j^{(\alpha)}}\,\underset{\xi=0}{\rm res}\,\frac{\nu_{jn\alpha}\prod_{i=1}^N(1-a_i\xi^{m_1})^{jn/m}}{\xi^{jnN_1+1}(1-z\xi^{m_1})}\,d\xi.
\end{equation*}
where $\xi:=\xi_{\infty_\alpha}$ is the local parameter at $\infty_\alpha$  given by \eqref{coordinates} and $\nu_{jn\alpha}$ as in \eqref{omegab}. Therefore, expanding every factor in Taylor series in $\xi\sim 0$,
\begin{align*}
\kappa_j&=&-\frac{m\,c_j}{jnA_j^{(\alpha)}}\,\underset{\xi=0}{\rm res}\,\frac{\nu_{jn\alpha}d\xi}{\xi^{jnN_1+1}}\sum_{k_1=0}^\infty {jn/m\choose k_1} (-a_1\xi^{m_1})^{k_1}\ldots
\sum_{k_N=0}^\infty {jn/m\choose k_N}(-a_N\xi^{m_1})^{k_N}\sum_{\q=0}^\infty(z\xi^{m_1})^\q  \\
 & = & -\frac{m\,c_j}{jnA_j^{(\alpha)}}\,\underset{\xi=0}{\rm res}\,\frac{\nu_{jn\alpha}d\xi}{\xi^{jnN_1+1}}\sum_{r=0}^\infty\Bigl[\sum_{k_1+\ldots+k_N+\q=r}(-1)^{r-\q}
				       {jn/m\choose k_1}\ldots{jn/m\choose k_N}a_1^{k_1}\ldots a_N^{k_N}z^\q\Bigr]\,\xi^{rm_1}.
\end{align*}
To obtain the residue we need to take the term of the series corresponding to the value of $r$ such that $jnN_1=rm_1$. Note that the constant $A_j^{(\alpha)}$ cancels out and we obtain $R_j$ \eqref{Rj}.
\end{proof}

\begin{example}
Let us set, for simplicity, $\theta_1=\theta_2=\frac{1}{4}$, $p=2$, and $ n=1,\;m=2,\;N=2$ so that $s=2$. In this case, $\beta_1=\frac{1}{4}$ and $\beta_2=-\frac{1}{4}$ and the polynomial in $a_1,\, a_2$ solution \eqref{polySchlesinger} to the Schlesinger system  is
\begin{equation*}
B^{(1)} =
\begin{pmatrix}
\frac{1}{4 } &  \frac{c_1}{2}(a_2-a_1)\\
0 & -\frac{1}{4}
\end{pmatrix}, \qquad
B^{(2)} =
\begin{pmatrix}
\frac{1}{4 } &  \frac{c_1}{2}(a_1-a_2)\\
0 & -\frac{1}{4}
\end{pmatrix}.
\end{equation*}
The solution given by Corollary \ref{cor_rat1} to the  corresponding Fuchsian system
  is
\begin{equation*}
    \Phi(z) = \begin{pmatrix}
       (z-a_1)^{\frac{1}{4}}(z-a_2)^{\frac{1}{4}}  & \frac{c_1(2z-a_1-a_2)}{(z-a_1)^{\frac{1}{4}}(z-a_2)^{\frac{1}{4}}} \\ \\
        0 & \frac{1}{(z-a_1)^{\frac{1}{4}}(z-a_2)^{\frac{1}{4}}} 
    \end{pmatrix}.
\end{equation*}
\end{example}

\subsection{Negative values of $n$}
\label{sect_rat2}

Following \cite{DragGontShram}, we chose $n<0$ coprime to $m$. For $n<0$, the differentials $\Omega_i^{(j)}(a)$ \eqref{diffmero} are not singular at the points $\infty_1, \dots, \infty_s$ at infinity and have poles at the finite ramification points $P_{a_1}, \dots, P_{a_N}$ of the coverings $\zeta:X_a\to\mathbb CP^1.$   As can be seen from  definition \eqref{diffmero} of the differentials, their behaviour at the pole $P=(\zeta,w)\sim P_{a_\nu}$ in terms of the standard local parameter $\xi_{a_\nu}$ \eqref{coordinates} is
\begin{equation}
\label{poleak}
\Omega_i^{(j)}=\frac{w^{-j|n|}}{\zeta-a_i}\,d\zeta=\frac m{\xi_{a_\nu}^{j|n|-m+1}(\xi_{a_\nu}^m+a_\nu-a_i)}\prod_{h=1,h\neq \nu}^N(\xi_{a_\nu}^m+a_\nu-a_h)^{-j|n|/m}\,d\xi_{a_\nu}.
\end{equation}
Thus, $\Omega_i^{(j)}$ has a pole of order $j|n|-m+1$ at $P_{a_\nu}$ with $\nu\neq i$ if $j|n|\geqslant m$ and of order $j|n|+1$ at $P_{a_i}.$ Note that, given that $n$ is coprime to $m$, the residues are nonzero only if $j$ is a multiple of $m$. Let us choose contours of integration in \eqref{bi} in the form 
\begin{equation}
\label{contour}
\gamma_{l-k}=\sum_{\nu=1}^N\frac{c^\nu_{l-k}}{2\pi\i}\gamma_{a_\nu}
\end{equation}
where $c^\nu_{l-k}\in\mathbb C$ are arbitrary constants and where $\gamma_{a_\nu}$ is a small contour encircling the point $P_{a_\nu}$. Then  we obtain from Theorem \ref{thmDGS} the following rational solutions to the Schlesinger system given in terms of the residues of $\Omega_i^{(j)}(a)$ at $P_{a_\nu}$
for some $\theta_i\in\mathbb C$ with $i=1, \dots, N$  (see Theorem 4 of \cite{DragGontShram}):
\begin{equation}
\label{ratSchlesinger}
B^{(i)}_{kl}=
\begin{dcases*}     
\sum_{\nu=1}^Nc^\nu_{l-k}\;\underset{P_{a_\nu}}{\rm res}\,\Omega_i^{(l-k)}(a),\quad\;\,  \text{if }  (l-k)/m\in\mathbb Z,  \;\text{ and } \; l>k \\ 
0,\qquad\qquad\qquad\qquad\qquad\; \text{if }  (l-k)/m\notin\mathbb Z \;\text{ or } \; l<k, \\
\theta_i - \frac{(k-1)n}{m}, \qquad\quad \qquad\,\text{if } \; l=k.
\end{dcases*} 
\end{equation}
Note, however, that $\sum_{\nu=1}^N\gamma_{a_\nu}$ corresponds to the sum of residues of the differentials $\Omega_i^{(j)}(a)$ and thus choosing all coefficients $c_{l-k}^\nu$ equal in \eqref{contour} leads to a trivial solution. The nonzero residues from \eqref{ratSchlesinger}, for $(l-k)/m\in\mathbb Z $, are explicitly given by \cite{DragGontShram}
\begin{equation*}
\underset{P_{a_\nu}}{\rm res}\,\Omega_i^{(l-k)}(a)=m\hspace{-1cm}\sum_{k_1+\ldots+k_N=(l-k)|n|/m-1}\frac{(-1)^{k_{\nu}}}{(a_{\nu}-a_i)^{k_{\nu}+1}}\prod_{h=1,h\ne\nu}^N\frac{{-(l-k)|n|/m\choose k_h}}{(a_{\nu}-a_h)^{k_h+(l-k)|n|/m}}
\end{equation*}
for $i=1,\ldots,N, \; i\ne\nu$ and
\begin{equation}
\label{rationsol2}
\underset{P_{a_\nu}}{\rm res}\,\Omega_{\nu}^{(l-k)}(a)=m\hspace{-1cm}\sum_{k_1+\ldots+\hat k_\nu+\ldots+k_N=(l-k)|n|/m}\prod_{h=1,h\ne\nu}^N\frac{{-(l-k)|n|/m\choose k_h}}{(a_{\nu}-a_h)^{k_h+(l-k)|n|/m}}
\end{equation}
where the hat symbol over $k_\nu$ indicates that the summation index $k_{\nu}$ is absent in the sum in \eqref{rationsol2}. Solutions obtained in Theorem \ref{cjtpiche} to the Fuchsian system \eqref{fuchsyst1} with rational coefficients \eqref{ratSchlesinger} are given by the next corollary.
\begin{corollary}
\label{cor_rat2}
 Let $n<0,$ and $p, m, N>0$ be integers such that ${\rm gcd}(n,m)=1$. Let  $\theta_i \in \mathbb{C}$ for $i=1, \dots, N$ and denote $\beta_k^{(i)}=\theta_i - \frac{(k-1)n}{m}$ for $k=1, \dots, p$. With $c_j^\nu\in\mathbb C$ as in \eqref{ratSchlesinger}, let $S_j=S_j(z)$ be the rational function 
\begin{equation}
\label{Sj}
S_j(z)=
\begin{dcases*}
\frac{-m^2}{jn}\sum_{\nu=1}^Nc^\nu_j\hspace{-0.3cm}\sum_{k_1+\ldots+k_N=j|n|/m-1}\frac{(-1)^{k_{\nu}}}{(a_{\nu}-z)^{k_{\nu}+1}}\prod_{h=1,h\ne\nu}^N\frac{{-j|n|/m\choose k_h}}{(a_{\nu}-a_h)^{k_h+j|n|/m}} ,\quad\text{if }  j/m\in\mathbb Z \\ 
\hspace{0.3cm}0,\qquad\qquad\qquad\qquad\qquad\qquad\qquad\qquad\qquad\qquad\qquad\qquad\qquad\qquad\quad\text{if }  j/m\notin\mathbb Z .
\end{dcases*}
\end{equation}
Then the following upper triangular matrix $\Phi(z)$ solves the Fuchsian system \eqref{fuchsyst1} with coefficients $B^{(i)},\;i=1, \dots, N,$ being the upper triangular $p \times p$ matrices given by \eqref{ratSchlesinger}:
\begin{equation*}
    \Phi = MD,
\end{equation*}
where
\begin{equation*}
    D = \mathrm{diag}\left(\prod_{i=1}^N(z-a_i)^{\beta_1^{(i)}}, \ldots, \prod_{i=1}^N(z-a_i)^{\beta_p^{(i)}}\right),
\end{equation*}
and
\begin{equation*}
    M_{kl} = \begin{dcases*}\sum_{q \, \vdash \, (l - k)}\,\prod_{j=1}^{l-k}\frac{S_j^{\sigma_q(j)}}{\sigma_q(j)!} \quad &\text{when } $l>k$,\\ 
    \quad 1 \quad &\text{when } $l= k$,\\ 
    \quad 0 \quad &\text{when } $l < k$,
    \end{dcases*}
\end{equation*}
where, as before,  $\sigma_q(j)$ is the number of repetitions of the part $j$ in the partition $q$.
\end{corollary}
\begin{proof}
Let $z\in\mathbb C\setminus\{a_1, \dots, a_N\}$. For $n<0$ and the compact Riemann surface $X_a$ corresponding to the superelliptic curve \eqref{curves}, the differentials $\Omega_i^{(j)}$ \eqref{diffmero} have poles at $P_{a_\nu}$ for $\nu=1, \dots, N$ as shown at the beginning of Section \ref{sect_rat2}. Similarly to the proof of Corollary \ref{cor_rat1}, we need to prove that $\kappa_j$ from Definition \ref{dfnkappa} for the choice of the integration contour $\gamma_j $ as in \eqref{contour} with $l-k=j$  coincides with $S_j$ \eqref{Sj}. We make an additional assumption that none of the contours $\gamma_{a_\nu}$ from \eqref{contour}  contains in its interior any point $P\in X_a$ for which $\zeta(P)=z$. With this choice, in terms of the standard local parameter $\xi_{a_\nu}$ \eqref{coordinates} at $P_{a_\nu}$, we have
\begin{equation*}
\kappa_j = -\frac{m}{jn}\sum_{\nu=1}^Nc^\nu_j\,\underset{P_{a_\nu}}{\rm res}\,\frac{w^{jn}\mathrm{d}\zeta}{\zeta - z} = \frac{-m^2}{jn}\sum_{\nu=1}^Nc^\nu_j\;\underset{\xi_{a_\nu}=0}{\rm res}\;\frac{\prod_{h=1,h\neq \nu}^N(\xi_{a_\nu}^m+a_\nu-a_h)^{-j|n|/m}\,d\xi_{a_\nu}}{\xi_{a_\nu}^{j|n|-m+1}(\xi_{a_\nu}^m+a_\nu-z)},
\end{equation*}
where the choice of the $m$th root should be made in the same way as in \eqref{poleak}. Expanding every factor in Taylor series at $\xi_{a_\nu}\sim 0$, we have
\begin{multline*}
\kappa_j =\frac{-m^2}{jn}\sum_{\nu=1}^Nc^\nu_j\underset{\xi_{a_\nu}=0}{\rm res}\;\frac{\Bigl(1+\frac{\xi_{a_\nu}^m}{a_{\nu}-z}\Bigr)^{-1}}{\xi_{a_\nu}^{j|n|-m+1}(a_{\nu}-z)}\prod_{h=1,h\ne\nu}^N\frac{d\xi_{a_\nu}}{(a_{\nu}-a_h)^{j|n|/m}}\Bigl(1+\frac{\xi_{a_\nu}^m}{a_{\nu}-a_h}\Bigr)^{-j|n|/m}\\
= \frac{-m^2}{jn}\sum_{\nu=1}^Nc^\nu_j\underset{\xi_{a_\nu}=0}{\rm res}\;\frac{d\xi_{a_\nu}}{\xi_{a_\nu}^{j|n|-m+1}}\sum_{r=0}^\infty\Bigl[\sum_{k_1+\ldots+k_N=r}\frac{(-1)^{k_{\nu}}}{(a_{\nu}-z)^{k_{\nu}+1}}\prod_{h=1,h\ne\nu}^N\frac{{-j|n|/m\choose k_h}}{(a_{\nu}-a_h)^{k_h+j|n|/m}}\Bigr]\,\xi_{a_\nu}^{rm}.
\end{multline*}
Now, to obtain the residue, for each $\nu$, we take the term of the series in $\xi_{a_\nu}$ with $r$ such that $rm=j|n|-m$ if such an integer $r$ exists. If $j|n|-m$ is not an integer multiple of $m$, which is equivalent to $j$ not being an integer multiple of $m$ due to the choice of $n$ coprime with $m$, the residue vanishes. This gives us $S_j$ from the statement of the corollary.
\end{proof}

\begin{example}
Let $p=3, \; m=N=2$, and $n=-1$. We thus consider the family of curves given by  equation $w^2=(\zeta-a_1)(\zeta-a_2)$ where the associated Riemann surfaces are of genus zero. 
Given that $1/m\notin\mathbb Z$, the first superdiagonals of $B^{(i)}$ and of $\Phi$ are zero due to \eqref{ratSchlesinger} and \eqref{Sj}, respectively. 
For  any $\theta_1, \theta_2, c\in\mathbb C$ and $c_2^1=0$, $c_2^2=c$,  Corollary \ref{cor_rat2} gives that the matrix
\begin{equation*}
\Phi=
\left(\begin{array}{ccc}
(z-a_1)^{\theta_1}(z-a_2)^{\theta_2} & 0 & \frac{2c}{a_1-a_2} (z-a_1)^{\theta_1+1}(z-a_2)^{\theta_2}
\\
0 & (z-a_1)^{\theta_1+\frac{1}{2}}(z-a_2)^{\theta_2+\frac{1}{2}} & 0 
\\
0 & 0 & (z-a_1)^{\theta_1+1}(z-a_2)^{\theta_2+1}
\end{array}\right)
\end{equation*}
satisfies the system
\begin{equation*}
\frac{d\Phi}{dz}=
\left(\begin{array}{ccc}
\frac{\theta_1}{z-a_1}+\frac{\theta_2}{z-a_2} & 0 & \frac{2c}{(a_1-a_2)^2} \left(\frac{1}{z-a_1}-\frac{1}{z-a_2}\right) 
\\
0 & \frac{\theta_1+\frac{1}{2}}{z-a_1}+\frac{\theta_2+\frac{1}{2}}{z-a_2}  & 0 
\\
0 & 0 & \frac{\theta_1+1}{z-a_1}+\frac{\theta_2+1}{z-a_2}
\end{array}\right)\Phi.
\end{equation*}

\end{example}

\bigskip
\bigskip
{\bf Acknowledgments.} The authors thank the anonymous referees for valuable suggestions and questions which greatly improved our presentation, especially in Section 4 on monodromy representation. 
BP would like to thank Maxime Fortier Bourque (Université de Montréal) for insightful discussions regarding sums over partitions of integers. BP also acknowledges that this research was partially funded by the Fonds de recherche du Québec - section Nature et Technologies (FRQ). The authors thank Renat Gontsov for careful reading of the manuscript and pointing out some errors of notation and terminology, which greatly improved the presentation.  VS gratefully acknowledges
support from the Natural Sciences and Engineering Research Council of Canada through a Discovery grant and from the University of Sherbrooke.

\end{document}